\documentclass[10pt,journal]{IEEEtran}
\usepackage{amsmath,amsfonts}
\usepackage{makecell}
\usepackage{tikz}
\usepackage{booktabs}
\usepackage{multirow}
\usepackage{amssymb}
\usepackage{subfigure}
\usepackage{algorithmic}
\usepackage{algorithm}
\usepackage{array}
\usepackage[caption=false,font=normalsize,labelfont=sf,textfont=sf]{subfig}
\usepackage{textcomp}
\usepackage{stfloats}
\usepackage{url}
\usepackage{color}
\usepackage{xcolor} 
\usepackage{verbatim}
\usepackage{graphicx}
\usepackage{cite}
\usepackage{amsthm}
\usepackage{bm}
\usepackage{enumitem}
\usepackage[T1]{fontenc}
\newtheorem{definition}{Definition}

\newtheorem{remark}{Remark}
\newtheorem{assumption}{Assumption}
\newtheorem{theorem}{Theorem}
\newtheorem{lemma}{Lemma}

\usetikzlibrary{positioning}
\definecolor{PalaceRed}{RGB}{190,30,45}
\definecolor{DeepSeaBlue}{RGB}{0,51,102}
\definecolor{ForestGreen}{RGB}{49,103,69}
\begin{document}

\title{Constrained coalition games for USV \\ swarm confrontation}

	\author{Cheng Yuwen, Jialing Zhou, \textit{Senior Member, IEEE}, Meng Luan, \textit{Member, IEEE}, \\
		Guanghui Wen, \textit{Senior Member, IEEE}, and Tingwen Huang, \textit{Fellow, IEEE}
	\thanks{Cheng Yuwen and Guanghui Wen are with the School of Automation, Southeast University, Nanjing 210096, China (e-mail: ywc@nuaa.edu.cn; ghwen@seu.edu.cn). 
		
		Jialing Zhou is with the School of Interdisciplinary Science, Beijing Institute of Technology, Beijing 100081, China (e-mail: jlzhou@bit.edu.cn).
		
		Meng Luan is with the School of Data Science, Lingnan University, Hong Kong SAR, China (e-mail: mengluan@ln.edu.hk).
		
		Tingwen Huang is with the Faculty of Computer Science and Control Engineering, Shenzhen University of Advanced Technology, Shenzhen 518055, China (e-mail: huangtingwen@suat-sz.edu.cn).}
}

\markboth{Manuscripts}%
{Shell \MakeLowercase{\textit{et al.}}: A Sample Article Using IEEEtran.cls for IEEE Journals}


\maketitle

\begin{abstract}
	This paper formulates unmanned surface vehicle (USV) swarm confrontation as a constrained coalition game, capturing cooperation within swarms and competition between swarms. Local task costs are aggregated into coalition objectives to account for tradeoffs among target approach, interception, and formation maintenance, subject to individual geographic constraints and supply line limits shared within each swarm. For prescribed formation sizes and task assignments, explicit linear matrix inequalities in the task weights are derived to establish strong monotonicity of the coalition pseudogradient, thereby guaranteeing the existence of a unique Nash equilibrium (NE) whenever the joint feasible set is nonempty. A distributed robust NE seeking strategy is then developed by integrating gradient tracking and action estimation with projected primal-dual dynamics using local information and neighbor communication. Adaptive compensation addresses unknown hydrodynamic parameters and bounded time-varying disturbances. The disturbance compensation gains are adjusted through local error feedback without requiring disturbance measurements, prior knowledge of disturbance bounds, or a prescribed disturbance generation model. Under the derived weight conditions, the stated dynamics and communication assumptions, and suitable gain conditions, the physical states of all USVs are shown to converge asymptotically to the unique feasible NE. Numerical simulations demonstrate the effectiveness of the proposed algorithm.
\end{abstract}
 
\section{Introduction}
\IEEEPARstart{C}{ooperation} is central to the mission effectiveness of unmanned surface vehicle (USV) swarms, enabling individual vessels to coordinate their actions toward collective objectives. Representative cooperative tasks include formation tracking \cite{Rasmus2021Formation}, leader escort \cite{Wen2024escort}, and coordinated task allocation \cite{MA2021227}. In maritime operations involving multiple USV swarms, the mission objectives of one swarm may directly conflict with those of another. Swarm confrontation thus combines cooperation within each swarm with competition between swarms \cite{Yuwen2025}. Each swarm's performance depends on both its own decisions and the actions of its opponent. Within each swarm, individual decisions affect the task performance of other members and jointly determine the collective outcome. These interactions call for a game formulation that treats each swarm as a competing decision maker while explicitly representing cooperation among its members.

Game theory provides a rigorous mathematical framework for characterizing such strategic interdependencies among rational decision makers \cite{rubinstein2007,weiss1999multiagent}. Building on the cooperative structure of distributed optimization \cite{Nedic2009} and the competitive structure of noncooperative games \cite{SALEHISADAGHIANI2016209}, coalition games capture cooperation within coalitions and competition between coalitions in a unified formulation \cite{YE2018266,huang2022,zhou2024a}. Specifically, each coalition consists of a group of agents and acts as a single player in a noncooperative game. The coalition cost function aggregates the local cost functions of its members, who cooperate to minimize this collective cost. This collective decision structure is reflected in USV swarm confrontation, where distinct formations within each swarm undertake tasks such as interception and defense in pursuit of a shared mission objective. Vessels cooperate within their respective formations to execute the assigned tasks and maintain formation geometry. The coalition cost aggregates individual vessel costs across these formations, accounting for tradeoffs among their task objectives. Within this coalition game, each USV determines its own actions according to the shared swarm objective, cooperating with teammates to balance individual task requirements while responding to competing swarms.

A key challenge in coalition games lies in seeking a Nash equilibrium (NE), namely a feasible joint action profile at which no coalition can reduce its cost through a feasible unilateral deviation \cite{Ye2019Unified,CHENG10156887,zou2023,nguyen2024,liu2024,ZENG201920,Deng9772719}. This requires each USV to assess how its decisions affect the collective cost of its swarm while accounting for the decisions of competing swarms. In practice, each USV evaluates its assigned tasks using locally available information, while communication limitations can restrict its access to other members' cost information and the actions of competing swarms. These information limitations motivate distributed NE seeking strategies that use neighbor communication to estimate coalition cost gradients and the actions of other agents. For networks with cooperative and competitive interactions, a distributed algorithm combining cost estimation and gradient updates was developed to seek an NE despite disconnected communication among cooperating agents \cite{CHENG10156887}. For games involving multiple coalitions, a distributed NE seeking algorithm was developed based on consensus protocols, with distributed observers estimating unavailable actions and designated agents facilitating communication between coalitions \cite{zou2023}. To accommodate directed communication within and between coalitions, a distributed NE seeking algorithm combining consensus and gradient tracking was developed \cite{nguyen2024}. Beyond information limitations, shared constraints introduce additional interdependence within each coalition, as one agent's action affects the feasible choices of other members. For coalition games with local feasible sets and nonlinear coupling inequality constraints, a distributed nonsmooth algorithm based on projected differential inclusions was developed for generalized NE seeking \cite{ZENG201920}. However, the algorithm assumes direct access to the actions of other agents needed for local subgradient evaluation. A distributed algorithm combining action estimation, subgradient descent, and projection operations was developed for nonsmooth coalition games with local convex constraints and linear coupling equality constraints within each coalition \cite{Deng9772719}. In USV swarm confrontation, shared supply line limits impose inequality constraints on the aggregate supply distance of each swarm. Unlike shared equalities, these inequalities require complementarity between nonnegative multipliers and aggregate constraint slacks at the NE, posing the challenge of satisfying this collective optimality condition through local updates under partial action information.

In the USV coalition game, the NE corresponds to a joint configuration of vessel positions and headings, whose realization depends on individual vessel dynamics. A neural predictor has been used to address USV model uncertainty within a distributed NE seeking scheme \cite{add_3}. Fuzzy prediction has also been employed to approximate unknown terms arising from uncertain nonlinearities and environmental disturbances \cite{add_4}. Combined with adaptive fault compensation, fuzzy prediction has enabled convergence to a neighborhood of the NE in the presence of model uncertainties, disturbances, and actuator faults \cite{add_5}. A distinct approach to disturbance attenuation uses a zero-sum differential game involving control inputs, sampling errors, and external disturbances to design robust control and event-triggered sampling policies for underactuated USVs \cite{add_6}. Within aggregative coalition games, auxiliary generalized NE seeking dynamics have been combined with adaptive tracking to establish convergence for uncertain Euler-Lagrange (EL) systems under local and coupling constraints \cite{huang2024a}. Structured disturbances have been addressed in aggregative games of linear systems through an adaptive internal model that accommodates constant and sinusoidal components with unknown amplitudes, frequencies, and phases \cite{LIU2026112603}. The internal model design nevertheless relies on a prescribed disturbance structure. Without such structural information or known disturbance bounds, the adequacy of fixed compensation gains cannot be guaranteed a priori. The resulting analytical difficulty lies in establishing sufficient disturbance rejection through local error feedback while ensuring bounded adaptive gains and asymptotic convergence of USV physical states to a feasible NE.

Motivated by the above discussion, this paper formulates USV swarm confrontation as a coalition game with local constraints on individual vessels and coupling inequality constraints within each swarm. A distributed robust NE seeking strategy is developed to accommodate uncertain USV dynamics and partial action information. Local and shared constraints are incorporated through projected primal-dual dynamics, while gradient tracking and action estimation provide the information required for coalition optimization through neighbor communication. The linear parameterization of the USV dynamics is exploited to construct adaptive compensation for unknown hydrodynamic parameters. Bounded disturbances are compensated by a signum term whose gain is adapted from local tracking errors, without requiring disturbance measurements, known magnitude bounds, or a prescribed generation model. Under the stated assumptions and suitable gain conditions, the physical states of all USVs are shown to converge asymptotically to the unique NE satisfying the local and shared constraints. Numerical simulations of USV swarm confrontation illustrate the theoretical results. The main contributions are summarized as follows.
\begin{enumerate}
	\item A constrained coalition game is formulated for USV swarm confrontation, capturing cooperation within swarms and competition between swarms. Local cost functions quantify individual task performance within distinct formations, while the coalition objective aggregates these costs to capture tradeoffs among tasks under individual geographic constraints and shared supply line limits. Unlike the formulation in \cite{huang2024a}, which imposes action consensus within each coalition, the proposed model accommodates distinct vessel positions and headings under a shared coalition objective.	
	\item Explicit linear matrix inequalities (LMIs) in the task weights are derived as sufficient conditions for NE existence and uniqueness. For prescribed formation sizes and task assignments, these conditions ensure positive definiteness of the symmetric part of the coalition pseudogradient Jacobian, thereby establishing strong monotonicity and guaranteeing a unique NE whenever the joint feasible set is nonempty. The analysis thus translates the strong monotonicity assumption in \cite{nguyen2024} into verifiable conditions on the task weights of the confrontation model.
	\item A distributed robust NE seeking strategy is developed for USVs with unknown hydrodynamic parameters and bounded disturbances under local and intra-coalition inequality constraints. Disturbance compensation gains are adapted from local error feedback without requiring disturbance measurements or prior knowledge of disturbance bounds. Compared with the general linear in \cite{Yuwen2025} and high-order integrator formulations in \cite{add_7}, respectively, the proposed strategy accommodates USV dynamics with unknown hydrodynamic parameters. Unlike the internal model approaches in \cite{Zhang8792368,LIU2026112603,Romano8727896}, disturbance compensation does not require a prescribed generation structure.
\end{enumerate}

This paper is organized as follows. Section \ref{Sec_Prelm} presents the preliminaries. Section \ref{Sec_SWARMCONUSV} formulates the USV swarm confrontation problem as a coalition game with local and coupling constraints. Section \ref{Sec_NE} establishes the existence and uniqueness of the NE and provides its KKT characterization. Section \ref{Sec_Strategy} develops a distributed robust NE seeking strategy and proves its convergence. Section \ref{Sec_SIMULATION} presents a numerical example to illustrate the theoretical results. Section \ref{Sec_Conclu} concludes the paper.

\textit{Notations}. Let $\mathbb{R}^n$ denote the set of $n$-dimensional real column vectors, $\mathbb{R}_+^n$ denote the set of nonnegative vectors in $\mathbb{R}^n$, and $\mathbb{R}^{m \times n}$ denote the set of $m \times n$ real matrices. $\mathbf{1}_n$ denotes an $n$-dimensional vector with all entries being one. $I_n$ represents the $n\times n$ identity matrix. Let $\mathbf{0}$ denote the zero vector or matrix of appropriate dimensions. The superscript $\top$ denotes the transpose. $\otimes$ and $\odot$ represent the Kronecker product and the Hadamard product, respectively. $\|x\|$ denotes the Euclidean norm of a vector $x$ and $\|X\|$ denotes the induced $2$-norm of a matrix $X$. For vectors $x_i$ or matrices $X_i$ with $i=1,\ldots,n$, $\operatorname{col}\{x_1,\ldots,x_n\}$ represents the column stacking of vectors $x_i$, and $\operatorname{diag}\{X_1,\ldots,X_n\}$ denotes a block-diagonal matrix with matrices $X_i$ as diagonal blocks. Symbol $\leq$ denotes element-wise inequality between vectors. For a symmetric matrix $P$, $P\succ \mathbf{0}$ indicates that $P$ is positive definite. $\lambda_{\min}(P)$ and $\lambda_{\max}(P)$ denote the minimum and maximum eigenvalues of $P$, respectively. For $x=[x_1,\ldots,x_n]^\top\in\mathbb{R}^n$, let
$P_+(x)=\operatorname{col}\{\max\{x_1,0\},\ldots,\max\{x_n,0\}\}$
denote the Euclidean projection of $x$ onto $\mathbb{R}_+^n$. For a real matrix $A=[a_{ij}]_{m\times n}$, $\operatorname{sgn}(A)=[\operatorname{sgn}(a_{ij})]_{m\times n}$ with $\operatorname{sgn}(\cdot)$ being the standard sign function.

\section{Preliminaries}\label{Sec_Prelm}

\subsection{Graph Theory}
A graph $\mathcal{G}=(\mathcal{V},\mathcal{E})$ consists of a vertex set $\mathcal{V}=\{\upsilon_1,\ldots,\upsilon_N\}$ and an edge set $\mathcal{E}$.
For an undirected graph, each edge is an unordered pair
$(\upsilon_i,\upsilon_j)$, indicating that $\upsilon_i$ and $\upsilon_j$
are connected. For a directed graph, each edge is an ordered pair
$(\upsilon_i,\upsilon_j)$, representing a directed connection from
$\upsilon_i$ to $\upsilon_j$. An undirected graph is connected if there
exists a path between any two distinct vertices. A directed graph is
strongly connected if there exists a directed path from $\upsilon_i$ to
$\upsilon_j$ for every pair of distinct vertices
$\upsilon_i,\upsilon_j\in\mathcal{V}$.

The adjacency matrix $ \mathcal{A}=[a_{ij}]_{N\times N} $ is defined as $a_{ij}=1$ if $(\upsilon_j, \upsilon_i) \in \mathcal{E}$, and $a_{ij}=0$ otherwise. Let $ \mathcal{D}=\operatorname{diag}\{d_1,\ldots,d_N\} $ with $ d_i=\sum_{j=1}^{N}a_{ij} $. The Laplacian matrix is given by $\mathcal{L} = \mathcal{D} - \mathcal{A}$. For an undirected graph, the Laplacian matrix is positive semi-definite.

\subsection{Projection Operator}
Let $\Omega \subset \mathbb{R}^m$ be a convex and closed set. A projection operator onto $\Omega$ is defined as $\mathcal{P}_{\Omega}(y) = \arg \min_{x \in \Omega} \|y - x\|$, $\forall y \in \mathbb{R}^m$.
To facilitate the subsequent analysis, the following properties of the projection operator are introduced.
\begin{lemma}[\hspace{-0.001cm}\cite{clarke2008nonsmooth}]\label{Projection_property}
	Define a function $V(x): \mathbb{R}^m \mapsto \mathbb{R}$ as $V(x) = \frac{1}{2} \left( \|x - \mathcal{P}_{\Omega}(\tilde{x})\|^2 - \|x - \mathcal{P}_{\Omega}(x)\|^2 \right)$, $\forall x, \tilde{x} \in \mathbb{R}^m$. Then the following statements hold:
	\begin{enumerate}[label=(\arabic*)]
		\item $(x - \mathcal{P}_{\Omega}(x))^{\top} (y - \mathcal{P}_{\Omega}(x)) \leq 0$, $\forall x \in \mathbb{R}^m,y \in \Omega$.
		\item $(\mathcal{P}_{\Omega}(x) - \mathcal{P}_{\Omega}(\tilde{x}))^{\top} (x - \tilde{x}) \geq \|\mathcal{P}_{\Omega}(x) - \mathcal{P}_{\Omega}(\tilde{x})\|^2$, $\forall x, \tilde{x} \in \mathbb{R}^m$.
		\item $V(x)$ is convex and differentiable with respect to $x$ with $\partial V(x) = \mathcal{P}_{\Omega}(x) - \mathcal{P}_{\Omega}(\tilde{x})$ and $V(x) \geq \frac{1}{2} \|\mathcal{P}_{\Omega}(x) - \mathcal{P}_{\Omega}(\tilde{x})\|^2$, $\forall x, \tilde{x} \in \mathbb{R}^m$.
	\end{enumerate}
\end{lemma}

\subsection{Euler-Lagrange Systems}

Consider an EL system described by
\begin{equation}\label{EL_Sys}
	E(x)\ddot{x}+C(x,\dot{x})\dot{x}+D(x)=u,
\end{equation}
where $u\in\mathbb{R}^{r}$ denotes the generalized
control input, $x\in\mathbb{R}^{r}$ denotes the generalized position
vector, and $\dot{x}$ and $\ddot{x}$ are the generalized
velocity and acceleration vectors, respectively.
The matrices $E(x),C(x,\dot{x})\in\mathbb{R}^{r\times r}$
are the inertia matrix and the Coriolis and centrifugal
matrix, respectively, and $D(x)\in\mathbb{R}^{r}$ is the
generalized gravity vector.

\begin{lemma}[\hspace{-0.001cm}\cite{Ortega1998}]\label{EL_property}
	The EL system \eqref{EL_Sys} possesses the following
	structural properties.
	\begin{enumerate}[label=(\arabic*)]
		\item\label{EL_PRO_1}
		The inertia matrix $E(x)$ is symmetric and satisfies
		$\underline{c}^{E}I_r\preceq E(x)
		\preceq\overline{c}^{E}I_r$
		for all $x\in\mathbb{R}^{r}$, where
		$\underline{c}^{E}$ and $\overline{c}^{E}$
		are positive constants.
		
		\item\label{EL_PRO_2}
		The matrix $\dot{E}(x)-2C(x,\dot{x})$ is
		skew-symmetric, where $\dot{E}(x)$ denotes
		the time derivative of $E(x(t))$ along a
		system trajectory.
		
		\item\label{EL_PRO_3} For any $\hat{y},\tilde{y}\in\mathbb{R}^{r}$, the dynamics can be linearly parameterized as $E(x)\hat{y}+C(x,\dot{x})\tilde{y}+D(x)=\Upsilon(x,\dot{x},\hat{y},\tilde{y})\mu$, where $\mu$ is a vector of unknown constant parameters and $\Upsilon(\cdot)$ is the corresponding known regressor matrix.
	\end{enumerate}
\end{lemma}

\section{USV Swarm Confrontation Model}\label{Sec_SWARMCONUSV}
Consider a swarm confrontation between the red and blue USVs with $\mathbb{N} = \{1, 2\}$. As shown in Fig. \ref{usvwar}, both sides involve multiple USVs, each departing from their respective coastal areas to carry out missions. The mission objectives are twofold. The first is to dispatch a group of USVs to deter the opposing swarm's command center. The second is to deploy defensive USVs to counter the opponent's intrusive USVs. Let $\mathbb{V}_1 = \{1, \ldots, m_1\}$ represent the set of USVs in red swarm. Furthermore, it holds that $\mathbb{V}_1 = \mathbb{V}_1^{\text{iv}} \cup \mathbb{V}_1^{\text{df}}$ where  $\mathbb{V}_1^{\text{iv}} = \{1, \ldots, m_1^{\text{iv}}\}$ is the set of intrusive USVs in the red swarm, and $\mathbb{V}_1^{\text{df}} = \{m_1^{\text{iv}}+1, \ldots, m_1\}$ is the set of defensive USVs. Similarly, define $\mathbb{V}_2 = \{1, \ldots, m_2\}$ as the set of USVs in blue swarm where $\mathbb{V}_2 = \mathbb{V}_2^{\text{iv}} \cup \mathbb{V}_2^{\text{df}}$, $\mathbb{V}_2^{\text{iv}} = \{1, \ldots, m_2^{\text{iv}}\}$ is the intrusive USV set of blue swarm, and $\mathbb{V}_2^{\text{df}} = \{m_2^{\text{iv}}+1, \ldots, m_2\}$ is the set of defensive USVs. For each $i\in\mathbb{N}$, the formation sizes $m_i^{\mathrm{iv}}$ and $m_i^{\mathrm{df}}$ are fixed positive integers satisfying
$m_i=m_i^{\mathrm{iv}}+m_i^{\mathrm{df}}$.
\begin{figure}[hbtp]
	\centering
	\includegraphics[width=0.5\textwidth]{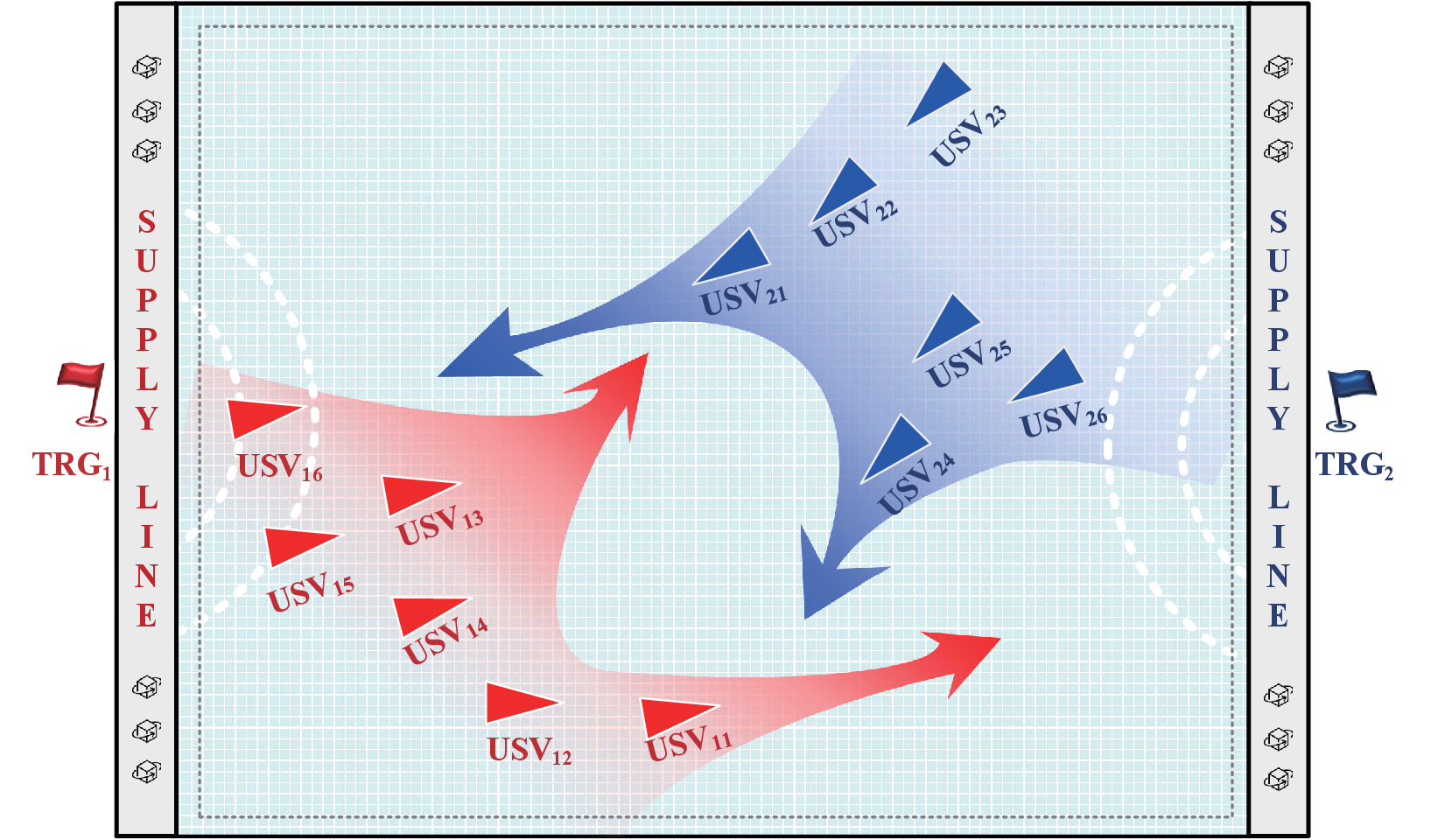}\\
	\vspace{-0.25cm}
	\caption{Confrontation scenario between red and blue USV swarms.}\label{usvwar}
\end{figure}

\subsection{USV Dynamics}
Following \cite{antonelli2018}, the dynamics of the $j$-th USV in the $i$-th swarm are described by
\begin{equation}\label{ch5_USV_dyn}
	\begin{aligned}
		\dot{x}_{ij} =& \Psi_{ij}(\phi_{ij})\nu_{ij},  \\
		M_{ij}\dot{\nu}_{ij} =& -\Pi_{ij}(\nu_{ij})\nu_{ij}-\Phi_{ij}\nu_{ij}+\tau_{ij}+d_{ij},
	\end{aligned}
\end{equation}
where $x_{ij} = \operatorname{col}\{x^{\text{usv}}_{ij}, y^{\text{usv}}_{ij},\phi_{ij}\}$ represents the position and orientation of the USV in the inertial frame, with $x^{\text{usv}}_{ij}$ and $y^{\text{usv}}_{ij}$ being the centroid position, and $\phi_{ij}$ the heading angle, $\nu_{ij} = \operatorname{col}\{\varrho_{ij}, \delta_{ij}, \theta_{ij}\}$ represents the velocity vector with $\varrho_{ij}$, $\delta_{ij}$, and $\theta_{ij}$ being the surge speed, sway speed, and yaw rate, respectively, $\tau_{ij} = \operatorname{col}\{\tau_{ij,\varrho}, \tau_{ij,\delta}, \tau_{ij,\theta}\}$ represents the control inputs with $\tau_{ij,\varrho}$, $\tau_{ij,\delta}$, and $\tau_{ij,\theta}$ being the surge thrust, sway thrust, and yaw moment, respectively. The vector $d_{ij}(t)\in\mathbb{R}^{3}$ denotes the time-varying external disturbance expressed in the body-fixed frame and satisfies $\|d_{ij}(t)\|\leq\tilde{d}_{ij}$ for all $t\geq0$, where $\tilde{d}_{ij}>0$ is an unknown constant. Furthermore, the matrices \(M_{ij}\), \(\Psi_{ij}(\phi_{ij})\), \(\Pi_{ij}(\nu_{ij})\), and \(\Phi_{ij}\) denote the inertia, rotation, Coriolis and centripetal, and hydrodynamic damping matrices, respectively, and are given by
\begin{align*}
	\Psi_{ij}  &=\left[\begin{array}{ccc}
		\cos \phi_{ij} & -\sin \phi_{ij} & 0 \\
		\sin \phi_{ij} & \cos \phi_{ij} & 0 \\
		0 & 0 & 1
	\end{array}\right]
\end{align*}
	\begin{align*}
	\Phi_{ij} &= \left[\begin{array}{ccc}
		-X_{\varrho} & 0 & 0 \\
		0 & -Y_{\delta} & -Y_{\theta} \\
		0 & -N_{\delta} & -N_{\theta}
	\end{array}\right],\\
	M_{ij}  &= \left[\begin{array}{ccc}
		\tilde{m}_{ij}-X_{\dot{\varrho}} & 0 & 0 \\
		0 &\tilde{m}_{ij}-Y_{\dot{\delta}} & \tilde{m}_{ij}x_{g}-Y_{\dot{\theta}} \\
		0 & \tilde{m}_{ij}x_{g}-N_{\dot{\delta}} & I_z-N_{\dot{\theta}}
	\end{array}\right],\\
	\Pi_{ij} &= \left[\begin{array}{ccc}
		0 & 0 & c_{\Pi_{ij}} \\
		0 & 0 & (\tilde{m}_{ij}-X_{\dot{\varrho}})\varrho_{ij} \\
		-c_{\Pi_{ij}} & -(\tilde{m}_{ij}-X_{\dot{\varrho}})\varrho_{ij} & 0
	\end{array}\right],
\end{align*}
where $c_{\Pi_{ij}}=-(\tilde{m}_{ij}-Y_{\dot{\delta}})\delta_{ij}-(\tilde{m}_{ij}x_g-Y_{\dot{\theta}})\theta_{ij}$, $\tilde{m}_{ij}$ represents the mass of USV, $x_{g}$ represents the distance from the center of mass of USV to the origin of the body frame, $I_z$ is the yaw moment of inertia, and $X_{*}$, $Y_{*}$, and $N_{*}$ represent the relevant hydrodynamic parameters.

All physical parameters are constant, and the hydrodynamic coefficients $X_*$, $Y_*$, and $N_*$ are unknown. The parameter values satisfy $Y_{\dot{\theta}}=N_{\dot{\delta}}$, $M_{ij}=M_{ij}^{\top}\succ0$, and $\Phi_{ij}+\Phi_{ij}^{\top}\succeq0$. Define the transformed matrices
\begin{align}
	E_{ij}(\phi_{ij}) &= \Psi_{ij}(\phi_{ij})M_{ij}\Psi_{ij}^{\top}(\phi_{ij}), \notag\\
	C_{ij}(\phi_{ij},\dot{x}_{ij}) &= \Psi_{ij}(\phi_{ij})\bigl(\Pi_{ij}-\mathcal{H}(\theta_{ij})\bigr)\Psi_{ij}^{\top}(\phi_{ij}), \notag\\
	D_{ij}(\phi_{ij}) &= \Psi_{ij}(\phi_{ij})\Phi_{ij}\Psi_{ij}^{\top}(\phi_{ij}), \label{TransMartixUSV}
\end{align}
where $\mathcal{H}(\theta_{ij})=M_{ij}\Psi_{ij}^{\top}\dot{\Psi}_{ij}$. Substituting $\nu_{ij}=\Psi_{ij}^{\top}(\phi_{ij})\dot{x}_{ij}$ into \eqref{ch5_USV_dyn} and premultiplying the dynamic equation by $\Psi_{ij}(\phi_{ij})$ yield the equivalent representation
\begin{equation}\label{USV_EL_Sys}
	\begin{aligned}
		&E_{ij}(\phi_{ij})\ddot{x}_{ij}+C_{ij}(\phi_{ij},\dot{x}_{ij})\dot{x}_{ij}+D_{ij}(\phi_{ij})\dot{x}_{ij}\\
		&\qquad=\Psi_{ij}(\phi_{ij})(\tau_{ij}+d_{ij}).
	\end{aligned}
\end{equation}
Under the stated parameter conditions, $E_{ij}$ and $C_{ij}$ satisfy properties \ref{EL_PRO_1} and \ref{EL_PRO_2} of Lemma \ref{EL_property}, while the transformed damping matrix satisfies $D_{ij}+D_{ij}^{\top}\succeq0$.

Define the unknown constant parameter vector $\mu_{ij}=\operatorname{col}\{\tilde{m}_{ij}-X_{\dot{\varrho}},\tilde{m}_{ij}-Y_{\dot{\delta}},\tilde{m}_{ij}x_g-Y_{\dot{\theta}},I_z-N_{\dot{\theta}},X_{\varrho},Y_{\delta},Y_{\theta},N_{\delta},N_{\theta}\}\in\mathbb{R}^{9}$. For any $\hat{y},\tilde{y}\in\mathbb{R}^{3}$, the transformed dynamics admit the linear parameterization
\begin{equation}\label{USV_LinearParameterization}
	\begin{aligned}
		&E_{ij}(\phi_{ij})\hat{y}+C_{ij}(\phi_{ij},\dot{x}_{ij})\tilde{y}+D_{ij}(\phi_{ij})\tilde{y}\\
		&\qquad=\Upsilon_{ij}(x_{ij},\dot{x}_{ij},\hat{y},\tilde{y})\mu_{ij},
	\end{aligned}
\end{equation}
where $\Upsilon_{ij}(x_{ij},\dot{x}_{ij},\hat{y},\tilde{y})\in\mathbb{R}^{3\times9}$ is a known continuous regressor matrix and is bounded on bounded sets of its arguments.

\subsection{Confrontation Objectives and Constraints}


Let the coordinates of the red swarm's command center be denoted by $x_{1}^{\text{p}_{\text{trg}}} = (x_1^{\text{trg}}, y_1^{\text{trg}})$, where $x_1^{\text{trg}}$ is the horizontal coordinate and $y_1^{\text{trg}}$ is the vertical coordinate. Similarly, let the coordinates of the blue swarm's command center be $x_{2}^{\text{p}_{\text{trg}}} = (x_2^{\text{trg}}, y_2^{\text{trg}})$, where $x_2^{\text{trg}}$ represents the horizontal coordinate and $y_2^{\text{trg}}$ represents the vertical coordinate.


Assume that there exists a sufficiently long supply line. Furthermore, the red and blue swarms are assumed to possess supply line coordinates represented by two vertically parallel lines. Specifically, the red side is positioned on the left at $x_{\text{R}}^\text{line}$, while the blue side is positioned on the right at $x^{\text{B}}_\text{line}$.

Define $x^{\text{p}}_{ij}=[x^{\text{usv}}_{ij}, y^{\text{usv}}_{ij}]^{\top}$. For both swarms, the task model of intrusive USV is described as follows.
\begin{enumerate}
	\item Assume that the $k$-th intrusive USV from the $i_1$-th swarm is tasked with deterring the command center of the $i_2$-th swarm where $ i_1, i_2 \in \mathbb{N}$ and $i_1 \neq i_2$. For $k \in \mathbb{V}_{i_1}^{\text{iv}}$, the specific cost function can be formulated as
	\begin{equation*}
		{J}^{\text{iv}}_1(x_{i_1k}^{\text{p}}, x_{i_2}^{\text{p}_{\text{trg}}}) =
		\left\|
		\begin{bmatrix}
			x_{i_1k}^{\text{usv}} \\
			y_{i_1k}^{\text{usv}}
		\end{bmatrix} -
		\begin{bmatrix}
			x_{i_2}^{\text{trg}} \\
			y_{i_2}^{\text{trg}}
		\end{bmatrix}
		\right\|^2.
	\end{equation*}
	
	\item While executing the intrusive task, assume that the $k$-th intrusive USV from the $i_1$-th swarm must avoid all USVs from the $i_2$-th swarm where $ i_1, i_2 \in \mathbb{N}$ and $i_1 \neq i_2$. For $k \in \mathbb{V}_{i_1}^{\text{iv}}$, the specific cost function is described as
	\begin{equation*}
		{J}^{\text{iv}}_2(x_{i_1k}^{\text{p}}, \mathbf{x}_{i_2}^{\text{p}}) =
		-\sum_{q \in \mathbb{V}_{i_2}} \left\|
		\begin{bmatrix}
			x_{i_1k}^{\text{usv}} \\
			y_{i_1k}^{\text{usv}}
		\end{bmatrix} -
		\begin{bmatrix}
			x_{i_2q}^{\text{usv}} \\
			y_{i_2q}^{\text{usv}}
		\end{bmatrix}
		\right\|^2, 
	\end{equation*}
	where $\mathbf{x}_{i_2}^{\text{p}}=\operatorname{col}\{x_{i_21}^{\text{p}},\ldots,x_{i_2m_{i_2}}^{\text{p}}\}$.
	
	\item Assume that the $k$-th intrusive USV from the $i$-th swarm is tasked with executing a pincer attack on the command center. For $ k \in \mathbb{V}_i^{\text{iv}}$ and $ i \in \mathbb{N}$, the specific cost function is designed as
	\[
	{J^{\text{iv}}_3}( \phi_{ik}) = \left\| \phi_{ik} - \tilde{\phi}_{ik}^{\text{usv}} \right\|^2, 
	\]
	where $\tilde{\phi}_{ik}^{\text{usv}}$ is the desired attacking angle, uniformly distributed on the intrusive side.
	
	\item Assume that the intrusive USVs from the $i$-th swarm are required to form a cohesive formation to concentrate their strength. For $q \in \mathbb{V}_i^{\text{iv}}$, the specific cost function is developed as
	\begin{equation*}
		{J^{\text{iv}}_4}(x_{iq}^{\text{p}}, \mathbf{x}_{i}^{\text{iv}}) = {\left\| \begin{bmatrix}
				x_{iq}^{\text{usv}} \\
				y_{iq}^{\text{usv}}
			\end{bmatrix} - \frac{1}{m_i^{\text{iv}}} \sum_{k \in \mathbb{V}_i^{\text{iv}}}
			\begin{bmatrix}
				x_{ik}^{\text{usv}} \\
				y_{ik}^{\text{usv}}
			\end{bmatrix} - \Delta_{iq} \right\|^2},
	\end{equation*}
	where $\mathbf{x}_i^{\text{iv}} = \operatorname{col}\{ x_{i1}^{\text{p}}, \ldots, x_{im_i^{\text{iv}}}^{\text{p}} \}$, and $\Delta_{iq}\in\mathbb{R}^{2}$ is a displacement relative to the formation center.
\end{enumerate}

Moreover, the task model of defensive USV is described as follows. For each $i\in\mathbb{N}$, let
$\sigma_i:\mathbb{V}_i^{\mathrm{df}}\to\mathbb{V}_{3-i}^{\mathrm{iv}}$ denote the prescribed task assignment. The same fixed assignment is used in the first three task costs defined below.
\begin{enumerate}
	\item To intercept the intrusive USV swarm, assume that each defensive USV in a given swarm is assigned to counter a specific intrusive USV from the opposing swarm. Specifically, the $k$-th defensive USV from the $i_1$-th swarm is tasked with defending against the $l$-th intrusive USV from the $i_2$-th swarm where $i_1, i_2 \in \mathbb{N}$ and $i_1 \neq i_2$. For $k \in \mathbb{V}_{i_1}^{\text{df}}$ and $l \in \mathbb{V}_{i_2}^{\text{iv}}$, the specific cost function is described as
	\begin{equation*}
		{J}^{\text{df}}_1(x_{i_1k}^{\text{p}}, x_{i_2l}^{\text{p}}) =
		\left\|
		\begin{bmatrix}
			x_{i_1k}^{\text{usv}} \\
			y_{i_1k}^{\text{usv}}
		\end{bmatrix} - \frac{1}{2}
		\begin{bmatrix}
			x_{i_2l}^{\text{usv}} + x_{i_1}^\text{trg} \\
			y_{i_2l}^{\text{usv}} + y_{i_1}^\text{trg}
		\end{bmatrix}
		\right\|^2.
	\end{equation*}
	
	\item To perform preemptive interception, assume that each defensive USV within a swarm must also maneuver proactively towards the enemy to establish tactical superiority. Let the $k$-th defensive USV from the $i_1$-th swarm be tasked with actively closing the distance to the $l$-th intrusive USV from the $i_2$-th swarm where $i_1, i_2 \in \mathbb{N}$ and $i_1 \neq i_2$. For $k \in \mathbb{V}_{i_1}^{\text{df}}$ and $l \in \mathbb{V}_{i_2}^{\text{iv}}$, the specific cost function is expressed as
	\begin{equation*}
		{J}^{\text{df}}_2(x_{i_1k}^{\text{p}}, x_{i_2l}^{\text{p}}) =
		\left\|
		\begin{bmatrix}
			x_{i_1k}^{\text{usv}} \\
			y_{i_1k}^{\text{usv}}
		\end{bmatrix} -
		\begin{bmatrix}
			x_{i_2l}^{\text{usv}} \\
			y_{i_2l}^{\text{usv}}
		\end{bmatrix}
		\right\|^2.
	\end{equation*}
	
	\item To facilitate an immediate response to the maneuvers and intercept the $l$-th intrusive USV from the $i_2$-th swarm, assume that the $k$-th defensive USV from the $i_1$-th swarm must orient itself at an angle opposite to that of the intrusive USV where $i_1, i_2 \in \mathbb{N}$ and $i_1 \neq i_2$. For $k \in \mathbb{V}_{i_1}^{\text{df}}$ and $l \in \mathbb{V}_{i_2}^{\text{iv}}$, the specific cost function is modeled as
	\begin{equation*}
		{J}^{\text{df}}_3(\phi_{i_1k}, \phi_{i_2l}) =
		\left\| \phi_{i_1k} - \phi_{i_2l}-\pi
		\right\|^2.
	\end{equation*}
	
	\item To capture the characteristics of a swarm, assume all defensive USVs should form a coordinated formation. For $q \in \mathbb{V}_{i}^{\text{df}}$, the specific cost function is formulated as
	\[
	{J}^{\text{df}}_4(x_{iq}^{\text{p}}, \mathbf{x}_{i,-q}^{\text{df}}) = {\left\| \begin{bmatrix}
			{x^{\text{usv}}_{{i}q}}\\
			{y^{\text{usv}}_{{i}q}}
		\end{bmatrix}-\frac{1}{m_i^{\text{df}}} {\sum_{k \in \mathbb{V}_i^{\text{df}}} { \begin{bmatrix}
					{x^{\text{usv}}_{{i}k}}\\
					{y^{\text{usv}}_{{i}k}}
		\end{bmatrix}  }   }-\Delta_{iq} \right\|^2},
	\]
	where $m_i^{\text{df}}=m_i-m_i^{\text{iv}}$, and $\mathbf{x}_{i,-q}^{\text{df}} = \operatorname{col} \{ {x}_{im_i^{\text{iv}}+1}^{\text{p}},$ $ {x}_{im_i^{\text{iv}}+2}^{\text{p}}, \ldots, {x}_{i(q-1)}^{\text{p}}, {x}_{i(q+1)}^{\text{p}}, \ldots, {x}_{im_i}^{\text{p}} \}$.
	
	\item To optimize the protection of the command center, assume that the desired position of the geometric center of the defensive USV formation coincides with the location of the command center. The specific cost function is modeled as
	\[
	{J}^{\text{df}}_5(\mathbf{x}_{i}^{\text{df}}, x_i^{\text{Ptrg}}) = {\left\| \frac{1}{m_i-m_i^{\text{iv}}} {\sum_{k \in \mathbb{V}_i^{\text{df}}} { \begin{bmatrix}
					{x^{\text{usv}}_{{i}k}}\\
					{y^{\text{usv}}_{{i}k}}
			\end{bmatrix}  } - \begin{bmatrix}
				{x_{i}^{\text{trg}}}\\
				{y_{i}^{\text{trg}}}
		\end{bmatrix}  } \right\|^2},
	\]
	where $\mathbf{x}_{i}^{\text{df}} = \operatorname{col} \{ {x}_{im_i^{\text{iv}}+1}^{\text{p}}, \ldots, {x}_{im_i}^{\text{p}} \}$.
\end{enumerate}

Beyond the aforementioned task objectives, all USVs operating within the battlefield are constrained by strict spatial boundaries. The operational domain is mathematically defined by the Cartesian coordinates $[X_\text{min}, X_\text{max}]$ and $[Y_\text{min}, Y_\text{max}]$, representing the horizontal and vertical boundaries of the battlefield, respectively. Consequently, these operational domains impose geographical constraints on all USVs in both the red and blue swarms, i.e.,
\begin{equation}\label{GeoConstraints}
	\begin{aligned}
		X_\text{min} \leq x_{ij}^{\text{usv}} \leq X_\text{max},\quad Y_\text{min} \leq y_{ij}^{\text{usv}} \leq Y_\text{max}, 
	\end{aligned}
\end{equation}

To quantify the logistical sustainability of the USV swarms on both red and blue sides, the supply line length is defined as the aggregate logistical distance from any USV position to the corresponding shoreline. This metric serves as a crucial parameter for maintaining effective supply chain management and operational continuity. Given the inherent logistical limitations, it is imperative that the frontlines of both red and blue USV formations maintain a strategic distance from the respective shorelines, ensuring the supply line length remains within prescribed operational thresholds. Therefore, the coupling constraints for maintaining the integrity of the supply lines for the swarms are formally expressed as
\begin{equation}\label{SuppConstraints}
	\begin{aligned}
		\sum_{j\in\mathbb{V}_1} (x_{1j}^{\text{usv}}-x^\text{line}_\text{R}) \le \sum_{j\in\mathbb{V}_1} \tilde{g}_{1j}, 
		\sum_{k\in\mathbb{V}_2} (x^\text{line}_\text{B}-x_{2k}^{\text{usv}}) \le \sum_{k\in\mathbb{V}_2} \tilde{g}_{2k}, 
	\end{aligned}
\end{equation}
where $\tilde{g}_{ij}$ is the logistical supply distance.

\subsection{Coalition Game Formulation}

The task objectives and constraints introduced above define a coalition game in which each swarm acts as a player. Let $N=2$ denote the number of swarms, $n=\sum_{i\in\mathbb{N}}m_i$ the total number of USVs, and $r=3$ the dimension of each USV action. The joint actions of swarm $i$ and all swarms are denoted by $\mathbf{x}_i=\operatorname{col}\{x_{i1},\ldots,x_{im_i}\}\in\mathbb{R}^{rm_i}$ and $\mathbf{x}=\operatorname{col}\{\mathbf{x}_1,\ldots,\mathbf{x}_N\}\in\mathbb{R}^{rn}$, respectively. For each $i\in\mathbb{N}$, let $\mathbf{x}_{-i}=\operatorname{col}\{\mathbf{x}_k:k\in\mathbb{N}\setminus\{i\}\}$ collect the actions of the other swarms in increasing index order. The full action profile is also written as $(\mathbf{x}_i,\mathbf{x}_{-i})$.

The local cost function $J_{ij}:\mathbb{R}^{rn}\to\mathbb{R}$ is a weighted sum of the task costs associated with the $j$-th USV in swarm $i$. For the red swarm, these functions are defined by
\begin{subequations}\label{CostFunctionRed}
	\begin{align}
		J_{1j}(\mathbf{x}) &= \sum_{l=1}^{4}\hat{l}^{\mathrm{R}}_{jl}J_l^{\mathrm{iv}}(\mathbf{x}), \quad j\in\mathbb{V}_1^{\mathrm{iv}}, \label{CostFunctionRed_a}\\
		J_{1j}(\mathbf{x}) &= \sum_{l=1}^{5}\hat{l}^{\mathrm{R}}_{jl}J_l^{\mathrm{df}}(\mathbf{x}), \quad j\in\mathbb{V}_1^{\mathrm{df}}. \label{CostFunctionRed_b}
	\end{align}
\end{subequations}
Similarly, the local cost functions of the blue swarm are
\begin{subequations}\label{CostFunctionBlue}
	\begin{align}
		J_{2j}(\mathbf{x}) &= \sum_{l=1}^{4}\hat{l}^{\mathrm{B}}_{jl}J_l^{\mathrm{iv}}(\mathbf{x}), \quad j\in\mathbb{V}_2^{\mathrm{iv}}, \label{CostFunctionBlue_a}\\
		J_{2j}(\mathbf{x}) &= \sum_{l=1}^{5}\hat{l}^{\mathrm{B}}_{jl}J_l^{\mathrm{df}}(\mathbf{x}), \quad j\in\mathbb{V}_2^{\mathrm{df}}. \label{CostFunctionBlue_b}
	\end{align}
\end{subequations}
Here, $\hat{l}^{\mathrm{R}}_{jl}$ and $\hat{l}^{\mathrm{B}}_{jl}$ are constant task weights. For each USV, $J_l^{\mathrm{iv}}(\mathbf{x})$ and $J_l^{\mathrm{df}}(\mathbf{x})$ denote the applicable task costs defined in the preceding subsection.

For each $i\in\mathbb{N}$ and $j\in\mathbb{V}_i$, the geographic constraints \eqref{GeoConstraints} define the local feasible set
\begin{equation}\label{LocalConstraints}
	\Omega_{ij} = \left\{x_{ij}\in\mathbb{R}^{r}\,\middle|\,B_{ij}x_{ij}\leq b_{ij}\right\},
\end{equation}
where $b_{ij}=[X_{\max},Y_{\max},-X_{\min},-Y_{\min}]^{\top}$ and
\begin{equation*}
	B_{ij} =
	\begin{bmatrix}
		1 & 0 & 0\\
		0 & 1 & 0\\
		-1 & 0 & 0\\
		0 & -1 & 0
	\end{bmatrix}.
\end{equation*}
The supply line constraints \eqref{SuppConstraints} define the coupling constraint set within swarm $i$ as
\begin{equation}\label{CoupledConstraints}
	\mathbb{X}_i = \left\{\mathbf{x}_i\in\mathbb{R}^{rm_i}\,\middle|\,\sum_{j=1}^{m_i}G_{ij}x_{ij}\leq\sum_{j=1}^{m_i}g_{ij}\right\},
\end{equation}
where $G_{1j}=[1,0,0]$ and $g_{1j}=\tilde{g}_{1j}+x_{\mathrm{R}}^{\mathrm{line}}$ for $j\in\mathbb{V}_1$, whereas $G_{2j}=[-1,0,0]$ and $g_{2j}=\tilde{g}_{2j}-x_{\mathrm{B}}^{\mathrm{line}}$ for $j\in\mathbb{V}_2$. Let $q$ and $p$ denote the numbers of local and coupling inequalities, respectively. For the present model, $q=4$ and $p=1$, with $B_{ij}\in\mathbb{R}^{q\times r}$, $b_{ij}\in\mathbb{R}^{q}$, $G_{ij}\in\mathbb{R}^{p\times r}$, and $g_{ij}\in\mathbb{R}^{p}$.

\begin{assumption}\label{ExistAss}
	For every $i\in\mathbb{N}$, there exists a joint action $\mathbf{x}_i$ such that $B_{ij}x_{ij}<b_{ij}$ for all $j\in\mathbb{V}_i$ and $\sum_{j=1}^{m_i}G_{ij}x_{ij}<\sum_{j=1}^{m_i}g_{ij}$.
\end{assumption}

Let $\bm{\Omega}_i=\prod_{j\in\mathbb{V}_i}\Omega_{ij}$. The feasible joint actions of swarm $i$ belong to $\bm{\Omega}_i\cap\mathbb{X}_i$. Its coalition cost is the average of the local costs, $J_i(\mathbf{x})=m_i^{-1}\sum_{j=1}^{m_i}J_{ij}(\mathbf{x})$. The members of each swarm cooperate to minimize this cost, given the actions of the competing swarms. The resulting coalition game, denoted by $\Game$, is formulated as
\begin{equation}\label{ch4_NCG}
	\begin{aligned}
		\min_{\mathbf{x}_i}\quad &J_i(\mathbf{x}_i,\mathbf{x}_{-i}), \quad i\in\mathbb{N},\\
		\text{s.t.}\quad &\mathbf{x}_i\in\bm{\Omega}_i\cap\mathbb{X}_i.
	\end{aligned}
\end{equation}

In game $\Game$, the $j$-th USV in swarm $i$ has direct access only to its local cost function $J_{ij}$, its own action $x_{ij}$, and the constraint parameters $B_{ij}$, $b_{ij}$, $G_{ij}$, and $g_{ij}$. Neighbor communication supports estimation of the coalition gradient $\nabla_{\mathbf{x}_i}J_i(\mathbf{x})$ and distributed updates of the variables $\lambda_{ij}$ associated with the coupling constraints. Communication within swarm $i$ is described by $\mathcal{G}_i=(\mathbb{V}_i,\mathcal{E}_i)$, with adjacency matrix $\mathcal{A}_i=[a^i_{jl}]_{m_i\times m_i}$ and Laplacian matrix $\mathcal{L}_i$.

The local cost functions may also depend on the actions of USVs outside the corresponding swarm, requiring action estimation across swarms. To describe the network used for this estimation, assign the $j$-th USV in swarm $i$ the index $\sum_{l=1}^{i-1}m_l+j$ in $\bar{\mathbb{N}}=\{1,\ldots,n\}$. The communication graph of all USVs is denoted by $\bar{\mathcal{G}}=(\bar{\mathbb{N}},\bar{\mathcal{E}})$, with adjacency matrix $\bar{\mathcal{A}}=[\bar{a}_{pq}]_{n\times n}$ and Laplacian matrix $\bar{\mathcal{L}}$.

\begin{assumption}\label{ch5_AssGraph}
	For every $i\in\mathbb{N}$, $\mathcal{G}_i$ is undirected and connected. The graph $\bar{\mathcal{G}}$ is directed, strongly connected, and weight-balanced.
\end{assumption}

\section{Nash Equilibrium Analysis}\label{Sec_NE}
This section derives sufficient conditions on the task weights for NE existence and uniqueness in coalition game \eqref{ch4_NCG}.

\begin{definition}\label{ch4_D_1}
	An action profile $\mathbf{x}^{\star}\in\bm{\Omega}\cap\mathbb{X}$, with $\bm{\Omega}=\prod_{i\in\mathbb{N}}\bm{\Omega}_i$ and $\mathbb{X}=\prod_{i\in\mathbb{N}}\mathbb{X}_i$, is an NE of $\Game$ if \[J_i(\mathbf{x}_i^{\star},\mathbf{x}_{-i}^{\star})\leq J_i(\mathbf{x}_i,\mathbf{x}_{-i}^{\star})\]
	for every $i\in\mathbb{N}$ and every $\mathbf{x}_i\in\bm{\Omega}_i\cap\mathbb{X}_i$.
\end{definition}

The analysis considers the prescribed formation sizes and fixed task assignments. The task weights, reference positions, formation offsets, and desired headings are constant. The heading variables in the quadratic task costs are treated as real-valued angles.

Define the coalition pseudogradient as $F(\mathbf{x})=\operatorname{col}\{\nabla_{\mathbf{x}_1}J_1(\mathbf{x}),\nabla_{\mathbf{x}_2}J_2(\mathbf{x})\}$. Since each coalition cost is the average of its local costs, $\nabla_{\mathbf{x}_i}J_i=m_i^{-1}\sum_{j\in\mathbb{V}_i}\nabla_{\mathbf{x}_i}J_{ij}$. Thus, the gradient associated with each coalition is obtained by differentiating the aggregate cost with respect to the joint action of all its members.

For $s\in\{\mathrm{iv},\mathrm{df}\}$ and each applicable task index $l$, define $W_l^{\mathrm{R},s}=\operatorname{diag}\{\hat{l}^{\mathrm{R}}_{jl}:j\in\mathbb{V}_1^s\}$ and $W_l^{\mathrm{B},s}=\operatorname{diag}\{\hat{l}^{\mathrm{B}}_{jl}:j\in\mathbb{V}_2^s\}$, with the diagonal entries ordered by increasing member index. Let $P_i^s=I_{m_i^s}-\mathbf{1}_{m_i^s}\mathbf{1}_{m_i^s}^{\top}/m_i^s$. The matrix $P_i^s\otimes I_2$ maps the stacked planar positions of the formation members to their deviations from the formation centroid. For $(i,c)\in\{(1,\mathrm{R}),(2,\mathrm{B})\}$, let $k=3-i$ denote the index of the opposing swarm. The prescribed assignment $\sigma_i$ is represented by $T_i\in\{0,1\}^{m_i^{\mathrm{df}}\times m_k^{\mathrm{iv}}}$, whose row $j-m_i^{\mathrm{iv}}$ has a unit entry in column $\sigma_i(j)$ and zeros elsewhere for each $j\in\mathbb{V}_i^{\mathrm{df}}$. Multiple rows may select the same column. Define
\begin{equation}\label{NE_PositionCurvature}
	\begin{aligned}
		A^{c,\mathrm{iv}} ={}& W_1^{c,\mathrm{iv}}-m_kW_2^{c,\mathrm{iv}}+P_i^{\mathrm{iv}}W_4^{c,\mathrm{iv}}P_i^{\mathrm{iv}},\\
		A^{c,\mathrm{df}} ={}& W_1^{c,\mathrm{df}}+W_2^{c,\mathrm{df}}+P_i^{\mathrm{df}}W_4^{c,\mathrm{df}}P_i^{\mathrm{df}}\\
		&+\frac{\mathbf{1}_{m_i^{\mathrm{df}}}^{\top}W_5^{c,\mathrm{df}}\mathbf{1}_{m_i^{\mathrm{df}}}}{(m_i^{\mathrm{df}})^2}\mathbf{1}_{m_i^{\mathrm{df}}}\mathbf{1}_{m_i^{\mathrm{df}}}^{\top}.
	\end{aligned}
\end{equation}
Set $A^{c,\mathrm{p}}=\operatorname{diag}\{A^{c,\mathrm{iv}},A^{c,\mathrm{df}}\}$. The cross derivatives between the two swarms are represented by
\begin{equation}\label{NE_PositionCross}
	B^{c,\mathrm{p}}=
	\begin{bmatrix}
		W_2^{c,\mathrm{iv}}\mathbf{1}_{m_i^{\mathrm{iv}}}\mathbf{1}_{m_k^{\mathrm{iv}}}^{\top}
		&
		W_2^{c,\mathrm{iv}}\mathbf{1}_{m_i^{\mathrm{iv}}}\mathbf{1}_{m_k^{\mathrm{df}}}^{\top}
		\\
		-\left(\frac{1}{2}W_1^{c,\mathrm{df}}+W_2^{c,\mathrm{df}}\right)T_i
		&
		\mathbf{0}_{m_i^{\mathrm{df}}\times m_k^{\mathrm{df}}}
	\end{bmatrix}.
\end{equation}
Consequently, $A^{c,\mathrm{p}}\in\mathbb{R}^{m_i\times m_i}$ and $B^{c,\mathrm{p}}\in\mathbb{R}^{m_i\times m_k}$.

Let $\mathbf{x}_i^{\mathrm{p}}=\operatorname{col}\{x_{ij}^{\mathrm{p}}:j\in\mathbb{V}_i\}$ denote the joint position vector of swarm $i$. For $s\in\{\mathrm{iv},\mathrm{df}\}$, define $\boldsymbol{\phi}_i^s=\operatorname{col}\{\phi_{ij}:j\in\mathbb{V}_i^s\}$ as the heading vector of formation $s$. The joint heading vector of swarm $i$ is $\boldsymbol{\phi}_i=\operatorname{col}\{\boldsymbol{\phi}_i^{\mathrm{iv}},\boldsymbol{\phi}_i^{\mathrm{df}}\}$. The prescribed heading vector is defined as $\tilde{\boldsymbol{\phi}}_i^{\mathrm{iv}}=\operatorname{col}\{\tilde{\phi}_{ij}^{\mathrm{usv}}:j\in\mathbb{V}_i^{\mathrm{iv}}\}\in\mathbb{R}^{m_i^{\mathrm{iv}}}$. Differentiating the coalition costs yields
\begin{equation}\label{NE_CoalitionGradients}
	\begin{aligned}
		\nabla_{\mathbf{x}_i^{\mathrm{p}}}J_i ={}& \frac{2}{m_i}\bigl((A^{c,\mathrm{p}}\otimes I_2)\mathbf{x}_i^{\mathrm{p}}+(B^{c,\mathrm{p}}\otimes I_2)\mathbf{x}_k^{\mathrm{p}}\bigr)+f_i^{\mathrm{p}},\\
		\nabla_{\boldsymbol{\phi}_i}J_i ={}& \frac{2}{m_i}
		\begin{bmatrix}
			W_3^{c,\mathrm{iv}}(\boldsymbol{\phi}_i^{\mathrm{iv}}-\tilde{\boldsymbol{\phi}}_i^{\mathrm{iv}})\\
			W_3^{c,\mathrm{df}}(\boldsymbol{\phi}_i^{\mathrm{df}}-T_i\boldsymbol{\phi}_k^{\mathrm{iv}}-\pi\mathbf{1}_{m_i^{\mathrm{df}}})
		\end{bmatrix},
	\end{aligned}
\end{equation}
where $f_i^{\mathrm{p}}\in\mathbb{R}^{2m_i}$ is independent of all actions and collects the reference-dependent terms.

The pseudogradient $F$ in \eqref{NE_CoalitionGradients} is affine, with a constant Jacobian whose symmetric part depends linearly on the task weights. Positive definiteness of this symmetric part ensures strong monotonicity, leading to the following LMI conditions.

\begin{assumption}\label{Ass_WeightLMI}
	All task weights in \eqref{CostFunctionRed} and \eqref{CostFunctionBlue} are positive. Let $\Xi_{\mathrm{g}}^{\mathrm{RB}}=m_1^{-1}B^{\mathrm{R},\mathrm{p}}+m_2^{-1}(B^{\mathrm{B},\mathrm{p}})^{\top}\in\mathbb{R}^{m_1\times m_2}$. The weights satisfy
	\begin{equation}\label{LMI_PositionNE}
		\Xi_{\mathrm{g}} \triangleq
		\begin{bmatrix}
			\frac{2}{m_1}A^{\mathrm{R},\mathrm{p}} & \Xi_{\mathrm{g}}^{\mathrm{RB}}\\
			(\Xi_{\mathrm{g}}^{\mathrm{RB}})^{\top} & \frac{2}{m_2}A^{\mathrm{B},\mathrm{p}}
		\end{bmatrix}
		\succeq\varepsilon_{\mathrm{p}}I_n
	\end{equation}
	for some $\varepsilon_{\mathrm{p}}>0$, together with
	\begin{equation}\label{LMI_HeadingNE}
		\begin{aligned}
			4m_2W_3^{\mathrm{R},\mathrm{iv}}-m_1T_2^{\top}W_3^{\mathrm{B},\mathrm{df}}T_2 &\succ0,\\
			4m_1W_3^{\mathrm{B},\mathrm{iv}}-m_2T_1^{\top}W_3^{\mathrm{R},\mathrm{df}}T_1 &\succ0.
		\end{aligned}
	\end{equation}
\end{assumption}

\begin{lemma}\label{Lemma_WeightFeasibility}
	For any fixed formation sizes and task assignments, there exist positive task weights satisfying Assumption \ref{Ass_WeightLMI}.
\end{lemma}

\begin{proof}
	Choose arbitrary positive task weights and denote the corresponding position matrix by $\Xi_{\mathrm{g}}^0$. For each $(i,c)\in\{(1,\mathrm{R}),(2,\mathrm{B})\}$, increase $W_1^{c,\mathrm{iv}}$, $W_4^{c,\mathrm{df}}$, and $W_5^{c,\mathrm{df}}$ by $(m_i\delta/2)I$ of appropriate dimensions, where $\delta>0$, while keeping the other position weights unchanged. Since $(P_i^{\mathrm{df}})^2=P_i^{\mathrm{df}}$ and $P_i^{\mathrm{df}}+\mathbf{1}_{m_i^{\mathrm{df}}}\mathbf{1}_{m_i^{\mathrm{df}}}^{\top}/m_i^{\mathrm{df}}=I_{m_i^{\mathrm{df}}}$, the resulting matrices satisfy $A^{c,\mathrm{p}}=A_0^{c,\mathrm{p}}+\frac{m_i\delta}{2}I_{m_i}$, $B^{c,\mathrm{p}}=B_0^{c,\mathrm{p}}$. Consequently, $\Xi_{\mathrm{g}}=\Xi_{\mathrm{g}}^0+\delta I_n\succ0$ whenever $\delta>\max\{0,-\lambda_{\min}(\Xi_{\mathrm{g}}^0)\}$. The heading conditions in \eqref{LMI_HeadingNE} can be satisfied by sufficiently increasing $W_3^{\mathrm{R},\mathrm{iv}}$ and $W_3^{\mathrm{B},\mathrm{iv}}$, respectively. All weights remain positive, completing the proof.
\end{proof}

\begin{theorem}\label{Thm_WeightNE}
	Under Assumptions \ref{ExistAss} and \ref{Ass_WeightLMI}, coalition game \eqref{ch4_NCG} admits a unique NE. Moreover, the pseudogradient $F$ is globally $\hbar$-strongly monotone for some $\hbar>0$, and the local gradient mappings $\nabla_{\mathbf{x}_i}J_{ij}$ are globally Lipschitz continuous with a common constant $\ell>0$.
\end{theorem}

\begin{proof}
	Let $\Xi_F=\frac{1}{2}\bigl(\nabla F+(\nabla F)^{\top}\bigr)$. Since the position and heading costs are separable, a suitable permutation of the action coordinates transforms $\Xi_F$ into a block diagonal matrix with position and heading blocks. By \eqref{NE_CoalitionGradients}, the position block is $\Xi_{\mathrm{g}}\otimes I_2$, which is positive definite under \eqref{LMI_PositionNE}. The heading block further decomposes into two blocks corresponding to $(\boldsymbol{\phi}_1^{\mathrm{iv}},\boldsymbol{\phi}_2^{\mathrm{df}})$ and $(\boldsymbol{\phi}_2^{\mathrm{iv}},\boldsymbol{\phi}_1^{\mathrm{df}})$. For the first block, the lower diagonal block is $2W_3^{\mathrm{B},\mathrm{df}}/m_2\succ0$, and its Schur complement is $2W_3^{\mathrm{R},\mathrm{iv}}/m_1-T_2^{\top}W_3^{\mathrm{B},\mathrm{df}}T_2/(2m_2)$. The first inequality in \eqref{LMI_HeadingNE} ensures that this Schur complement is positive definite. The same argument, with the red and blue indices interchanged, establishes positive definiteness of the second block under the second inequality in \eqref{LMI_HeadingNE}. Hence, $\Xi_F\succ0$.
	
	With $\hbar=\lambda_{\min}(\Xi_F)>0$, the affine structure of $F$ gives
	\begin{equation}\label{NE_StrongMonotonicity}
			(\mathbf{x}-\mathbf{y})^{\top}\bigl(F(\mathbf{x})-F(\mathbf{y})\bigr)\geq\hbar\|\mathbf{x}-\mathbf{y}\|^2
	\end{equation}
	for all $\mathbf{x},\mathbf{y}\in\mathbb{R}^{3n}$, establishing global strong monotonicity. The quadratic structure of the local costs also implies that all mappings $\nabla_{\mathbf{x}_i}J_{ij}$ are affine and therefore globally Lipschitz continuous. Since the number of local costs is finite, their Lipschitz constants admit a common upper bound $\ell>0$. The pseudogradient $F$ is likewise globally Lipschitz continuous.
	
	Assumption \ref{ExistAss} and the affine constraints ensure that $\bm{\Omega}\cap\mathbb{X}$ is nonempty, closed, and convex. Applying \eqref{NE_StrongMonotonicity} to action profiles that differ only in $\mathbf{x}_i$ establishes strong monotonicity of $\nabla_{\mathbf{x}_i}J_i(\cdot,\mathbf{x}_{-i})$ and hence strong convexity of $J_i(\cdot,\mathbf{x}_{-i})$. Since each coalition's feasible set is independent of the opposing actions, the first-order optimality conditions for all coalitions are equivalent to the variational inequality $(\mathbf{x}-\mathbf{x}^{\star})^{\top}F(\mathbf{x}^{\star})\geq0$ for every $\mathbf{x}\in\bm{\Omega}\cap\mathbb{X}$, with $\mathbf{x}^{\star}\in\bm{\Omega}\cap\mathbb{X}$.
	
	Let $L_F>0$ be a global Lipschitz constant of $F$, and choose $\varepsilon\in(0,2\hbar/L_F^2)$. The nonexpansiveness of the Euclidean projection, together with strong monotonicity and Lipschitz continuity of $F$, implies that $\mathbf{x}\mapsto\mathcal{P}_{\bm{\Omega}\cap\mathbb{X}}(\mathbf{x}-\varepsilon F(\mathbf{x}))$ is a contraction from $\bm{\Omega}\cap\mathbb{X}$ into itself, with contraction factor at most $\sqrt{1-2\varepsilon\hbar+\varepsilon^2L_F^2}<1$. Since this set is complete, Banach's fixed point theorem guarantees a unique fixed point. By the characterization of Euclidean projection, this fixed point is the unique solution of the variational inequality and therefore the unique NE of \eqref{ch4_NCG}.
\end{proof}
\begin{remark}
	The decomposition in \eqref{NE_PositionCurvature} characterizes the effects of the task weights on the position Hessian. The term $-m_kW_2^{c,\mathrm{iv}}$ in \eqref{NE_PositionCurvature} contributes negative curvature proportional to the opposing swarm size. Since $P_i^{\mathrm{iv}}\mathbf{1}=\mathbf{0}$, the formation term vanishes along uniform translation directions. Positive curvature in these directions therefore requires the target-approach contribution $W_1^{c,\mathrm{iv}}$ to dominate the negative contribution $m_kW_2^{c,\mathrm{iv}}$.
\end{remark}


The NE is further characterized by the Karush-Kuhn-Tucker conditions of the coalition problems.

\begin{lemma}\label{Lemma_NEKKT}
	Under the conditions of Theorem \ref{Thm_WeightNE}, an action profile $\mathbf{x}^{\star}$ is the NE of $\Game$ if and only if, for every $i\in\mathbb{N}$, there exist $\omega_{ij}^{\star}\in\mathbb{R}_{+}^{q}$ and $\lambda_{ij}^{\star}\in\mathbb{R}^{p}$, $j\in\mathbb{V}_i$, such that $\mathcal{P}_{+}(\lambda_{i1}^{\star})=\cdots=\mathcal{P}_{+}(\lambda_{im_i}^{\star})$ and, for every $j\in\mathbb{V}_i$,
	\begin{subequations}\label{constraints_ij}
		\begin{align}
			\frac{\partial J_i(\mathbf{x}^{\star})}{\partial x_{ij}}+G_{ij}^{\top}\mathcal{P}_{+}(\lambda_{ij}^{\star})+B_{ij}^{\top}\omega_{ij}^{\star} &= \mathbf{0}, \label{constraint_ija}\\
			\sum_{l=1}^{m_i}G_{il}x_{il}^{\star} &\leq \sum_{l=1}^{m_i}g_{il}, \label{constraint_ijb}\\
			B_{ij}x_{ij}^{\star} &\leq b_{ij}, \label{constraint_ijc}\\
			\mathcal{P}_{+}(\lambda_{ij}^{\star})\odot\left(\sum_{l=1}^{m_i}G_{il}x_{il}^{\star}-\sum_{l=1}^{m_i}g_{il}\right) &= \mathbf{0}, \label{constraint_ije}\\
			\omega_{ij}^{\star}\odot(B_{ij}x_{ij}^{\star}-b_{ij}) &= \mathbf{0}. \label{constraint_ijf}
		\end{align}
	\end{subequations}
\end{lemma}

%

\section{Distributed Robust Control Strategy}\label{Sec_Strategy}
This section develops a distributed robust NE seeking strategy for uncertain USV dynamics and establishes the convergence of the physical states to the unique NE of \eqref{ch4_NCG}.

\subsection{Strategy Design}
For the $j$-th USV in the $i$-th swarm, the controller $\tau_{ij}$ is given by
\begin{subequations}\label{USV_controller}
	\begin{align}
		\tau_{ij} =& -\gamma \Psi_{ij}^{\top}(\phi_{ij}) e_{ij} + \Psi_{ij}^{\top}(\phi_{ij}) \Upsilon_{ij}(x_{ij}, \dot{x}_{ij}, \ddot{\hat{x}}_{ij}, \dot{\hat{x}}_{ij})\hat{\mu}_{ij} \notag\\
		&- \Psi_{ij}^{\top}(\phi_{ij})\operatorname{sgn}(e_{ij})\hat{d}_{ij}, \label{USV_controller_a}\\
		\dot{\hat{\mu}}_{ij} =& -\Upsilon^{\top}_{ij}(x_{ij}, \dot{x}_{ij}, \ddot{\hat{x}}_{ij}, \dot{\hat{x}}_{ij}) e_{ij}, \label{USV_controller_b}\\
		\dot{\hat{d}}_{ij} =& e_{ij}^{\top} \operatorname{sgn}(e_{ij}), \label{USV_controller_c}\\
		e_{ij} =& \dot{x}_{ij} - \dot{\hat{x}}_{ij}, \label{USV_controllerd} \\
		\dot{\hat{x}}_{ij} =& \vartheta_{ij} - (x_{ij} - \eta_{ij}). \label{USV_controllere}
	\end{align}
\end{subequations}
where $\gamma$ is a positive constant gain.

For all $i \in \mathbb{N}$ and $j \in \mathbb{V}_{i}$, the auxiliary variables $\vartheta_{ij}$ and $\eta_{ij}$ are governed by
\begin{subequations}\label{AuxSys}
	\begin{align}
		\dot{\eta}_{ij} =& \vartheta_{ij}, \\
		\dot{\vartheta}_{ij} =& -\alpha \vartheta_{ij} - \big( R_{ij} \xi_{ij} + G_{ij}^{\top} \mathcal{P}_{+}(\lambda_{ij}) \notag\\
		&+ B_{ij}^{\top}\mathcal{P}_{+}(\omega_{ij} + B_{ij}\eta_{ij} - b_{ij})\big),
	\end{align}
\end{subequations}
where $\alpha$ is a positive constant gain, $\xi_{ij}=\operatorname{col}\{\xi_{ij1},\ldots,$ $\xi_{ijm_i}\}$, $R_{ij} = \operatorname{row}\{R_{ij1},\ldots, R_{ijm_i}\}$ with $R_{ijk} = I_r$ if $k = j$ and $R_{ijk} = \mathbf{0}$ otherwise.

For all $i \in \mathbb{N}$ and $j \in \mathbb{V}_{i}$, the Lagrange multipliers $\lambda_{ij}$ and $\omega_{ij}$ are updated by
\begin{subequations}\label{LagMulti}
	\begin{align}
		\dot{\omega}_{ij} =& B_{ij}\vartheta_{ij} - \omega_{ij} 
		+ \mathcal{P}_{+}(\omega_{ij} + B_{ij}\eta_{ij} - b_{ij}), \\
		\dot{\lambda}_{ij} =& -\lambda_{ij} + \mathcal{P}_{+}(\lambda_{ij}) 
		+ G_{ij} \eta_{ij} - g_{ij} 
		- \sum_{l=1}^{m_i} a^{i}_{jl} (\rho_{ij} - \rho_{il}) \notag\\
		&- \sum_{l=1}^{m_i} a^{i}_{jl} 
		(\mathcal{P}_{+}(\lambda_{ij}) - \mathcal{P}_{+}(\lambda_{il}))
		+ G_{ij}\vartheta_{ij},\\
		\dot{\rho}_{ij} =& \sum_{l=1}^{m_i} a^{i}_{jl}
		(\mathcal{P}_{+}(\lambda_{ij}) - \mathcal{P}_{+}(\lambda_{il})).
	\end{align}
\end{subequations}
where $\omega_{ij}$ is the Lagrange multiplier for the local constraint $\Omega_{ij}$, and $\mathcal{P}_{+}(\lambda_{ij})$ represents the Lagrange multiplier for the coupling constraint $\mathbb{X}_{i}$.

For all $i \in \mathbb{N}$, $j \in \mathbb{V}_{i}$, and $k \in \mathbb{V}_{i}$, $\xi_{ijk}$ is utilized to track the coalition gradient and is generated by
\begin{subequations}\label{TrackGrad}
	\begin{align}
		\dot{\xi}_{ijk} =& -\beta \Big( \xi_{ijk} + \sum_{l=1}^{m_i} a^{i}_{jl} (\xi_{ijk} - \xi_{ilk}) \notag \\
		&+ \sum_{l=1}^{m_i} a^{i}_{jl} (\zeta_{ijk} - \zeta_{ilk}) - \frac{\partial J_{ij}(\chi_{ij})}{\partial x_{ik}} \Big), \\
		\dot{\zeta}_{ijk} =& \beta \sum_{l=1}^{m_i} a^{i}_{jl} (\xi_{ijk} - \xi_{ilk}),
	\end{align}
\end{subequations}
where $\beta$ is a positive constant gain, $\chi_{ij} = \operatorname{col}\{s_{\sum_{l=1}^{i-1} m_l + j,1},$ $\ldots, s_{\sum_{l=1}^{i-1} m_l + j, n}\} \in \mathbb{R}^{rn}$ represents the estimate of all actions by the $j$-th agent of the $i$-th coalition, $s_{\sum_{l=1}^{i-1} m_l + j, l}$ denotes the estimate of the auxiliary variable $\eta_{ij}$ for the reindexed $l$-th agent, $l \in \bar{\mathbb{N}}$, and $\eta_{ij}$ is reindexed as $\varsigma_{\sum_{l=1}^{i-1} m_l + j}$, i.e., $\eta_{ij} \triangleq \varsigma_{\sum_{l=1}^{i-1} m_l + j}$.

For all $ p, l \in \bar{\mathbb{N}} $, $ s_{p,l} $ denotes the estimate by the $ p $-th agent of the auxiliary variable $ \varsigma_l $ of the $ l $-th agent after renumbering. Let $ s_{p,-p} = \operatorname{col}\{s_{p,1}, s_{p,2}, \ldots, s_{p,p-1}, s_{p,p+1}, \ldots, s_{p,n}\} \in \mathbb{R}^{r(n-1)} $, and it is governed by
\begin{align}\label{EstAct}
	\dot{s}_{p,-p} =& -\kappa \sum_{l=1}^{n} \bar{a}_{pl} (s_{p,-p} - s_{l,-p}), 
\end{align}
where $\kappa$ is a positive constant gain, $s_{l,-p} = \operatorname{col}\{s_{l,1},s_{l,2}, \ldots,$ $s_{l,p-1}, s_{l,p+1}, \ldots, s_{l,n}\}$, and $s_{l,l} = \varsigma_{l}$.

\begin{remark}
	The proposed algorithm comprises five interconnected parts: the control law \eqref{USV_controller}, the auxiliary system \eqref{AuxSys}, the Lagrange multiplier system \eqref{LagMulti}, the gradient tracking system \eqref{TrackGrad}, and the action estimate system \eqref{EstAct}.
\end{remark}

\subsection{Convergence Analysis}
For all $p\in\bar{\mathbb{N}}$, define
\begin{align*}
	\Theta_p =& \big[\mathbf{0}^{\top}_{p-1},1,\mathbf{0}^{\top}_{n-p}\big]\otimes I_{r} \in \mathbb{R}^{r \times nr},\\
	 \Xi_p=&\left[ {\begin{array}{*{20}{c}}
			{{I_{p - 1}}}&{\bf{0}}&{\bf{0}}\\
			{\bf{0}}&{\bf{0}}&{{I_{n-p}}}
	\end{array}} \right]\otimes I_{r} \in \mathbb{R}^{(n-1)r \times nr}.
\end{align*}

Define $\Xi = \operatorname{diag}\{\Xi_1,\ldots,\Xi_n\} \in \mathbb{R}^{n(n-1)r\times n^2r}$ and $\Theta = \operatorname{diag}\{\Theta_1,\ldots,\Theta_n\} \in \mathbb{R}^{nr\times n^2r}$. It can be verified by some algebraic manipulations that $\Theta\Theta^{\top} = I_{nr}$, $\Theta\Xi^{\top} = \mathbf{0}$, $\Xi \Xi^{\top} = I_{n(n-1)r\times n(n-1)r}$, $\Xi \Theta^{\top} = \mathbf{0}$ and $\Theta^{\top}\Theta + \Xi^{\top}\Xi = I_{n^2r}$. Then, one can get
\begin{equation}\label{ExtS_Relation}
	\eta = \Theta s, \quad z = \Xi s, \quad s = \Theta^{\top} \eta + \Xi^{\top} z,
\end{equation}
where $\eta  = \operatorname{col}\{\eta_{1}, \ldots, \eta_{N}\}$, $\eta_i = \operatorname{col}\{\eta_{i1}, \ldots, \eta_{im_i}\}$, $s=\operatorname{col}\{s_1,\ldots,s_n\}$, $s_{l}=(\varsigma_{l},s_{l,-l})$, $\forall l\in\bar{\mathbb{N}}$, and $z = \operatorname{col}\{s_{1,-1}, \ldots, s_{n,-n}\}$.

Let $\bar{z} = z - \Xi(\mathbf{1}_n \otimes\eta )$. It follows from \eqref{ExtS_Relation} that
\[
\Xi^{\top} \Xi (\mathbf{1}_n \otimes \eta) + \Theta^{\top} \eta = \mathbf{1}_n \otimes \eta,\quad
s = \Xi^{\top} \bar{z} + \mathbf{1}_n \otimes \eta.
\]
Moreover, define $\mathbf{L} = \bar{\mathcal{L}} \otimes I_{nr}$. With \eqref{ExtS_Relation}, system \eqref{EstAct} can be reformulated as
\begin{equation*}
	\dot{z} = -\kappa \Xi \mathbf{L} (\Xi^{\top} z + \Theta^{\top} \eta).
\end{equation*}
By the definition of $\bar{z}$, one can obtain
\begin{equation}\label{s_Sys}
	\begin{aligned}
		\dot{\bar{z}} =& \dot{z} - \Xi(\mathbf{1}_n \otimes \dot{\eta}) \\
		=& -\kappa \Xi \mathbf{L} \Xi^{\top} \bar{z} - \Xi(\mathbf{1}_n \otimes \dot{\eta}) \\
		=& -\kappa \Xi \mathbf{L} \Xi^{\top} \bar{z} - \Xi(\mathbf{1}_n \otimes \vartheta ).
	\end{aligned}
\end{equation}

\begin{lemma}[\hspace{-0.001cm}\cite{Gadjov8354898}]\label{Lemma_Graph2}
	Under Assumption \ref{ch5_AssGraph}, $\Xi(\mathbf{L}+\mathbf{L}^{\top})\Xi^{\top}$ is a positive definite matrix.
\end{lemma}

Substituting controller \eqref{USV_controller} into system \eqref{USV_EL_Sys}, the vector form of the agents in the $i$-th coalition is given by
\begin{subequations}\label{ELSys_i}
	\begin{align}
	  & E_{i}(\mathbf{x}_{i}) \ddot{\mathbf{x}}_{i} + C_{i}(\mathbf{x}_{i},\dot{\mathbf{x}}_{i}) \dot{\mathbf{x}}_{i} + D_{i}(\mathbf{x}_{i})\dot{\mathbf{x}}_{i} \notag\\
		=& \hspace{-0.05cm}- \hspace{-0.05cm}\gamma e_{i} \hspace{-0.05cm} + \hspace{-0.05cm} \Upsilon_{i}(\mathbf{x}_{i}, \dot{\mathbf{x}}_{i}, \ddot{\hat{\mathbf{x}}}_{i}, \dot{\hat{\mathbf{x}}}_{i})\hat{\mu}_{i}  \hspace{-0.05cm}-  \hspace{-0.05cm} \operatorname{sgn}(\varepsilon_{i})\hat{d}_{i}  \hspace{-0.05cm}+ \hspace{-0.05cm} \Psi_{i}d_{i}, \\
		\dot{\hat{\mu}}_{i} =& -\Upsilon^{\top}_{i}(\mathbf{x}_{i}, \dot{\mathbf{x}}_{i}, \ddot{\hat{\mathbf{x}}}_{i}, \dot{\hat{\mathbf{x}}}_{i}) e_{i}, \\
		\dot{\hat{d}}_{i} =& \operatorname{sgn}(\varepsilon_{i}^{\top})e_i, \\
		e_{i} =& \dot{\mathbf{x}}_{i} - \dot{\hat{\mathbf{x}}}_{i}, \\
		\dot{\hat{\mathbf{x}}}_{i} =& \vartheta_{i} - (\mathbf{x}_{i} - \eta_{i}), 
	\end{align}
\end{subequations}
where $E_i(\mathbf{x}_{i})=\operatorname{diag}\{E_{i1}(x_{i1}),\ldots,E_{im_i}(x_{im_i})\}$,  $e_i=\operatorname{col}\{e_{i1},\ldots,e_{im_i}\}$, $\hat{\mu}_i=\operatorname{col}\{\hat{\mu}_{i1},\ldots,\hat{\mu}_{im_i}\}$,  $\varepsilon_{i}=\operatorname{diag}\{e_{i1},\ldots,e_{im_i}\}$, $\hat{d}_i=\operatorname{col}\{\hat{d}_{i1},\ldots,\hat{d}_{im_i}\}$, $\Psi_{i}=\operatorname{diag}\{\Psi_{i1}(\phi_{i1}),\ldots,\Psi_{im_i}(\phi_{im_i})\}$, $d_i=\operatorname{col}\{d_{i1},\ldots,d_{im_i}\}$, $\eta_i=\operatorname{col}\{\eta_{i1},\ldots,\eta_{im_i}\}$,  $D_{i}(\mathbf{x}_{i})=\operatorname{diag}\{D_{i1}(x_{i1}),\ldots,D_{im_i}(x_{im_i})\}$, $\Upsilon_{i}(\mathbf{x}_{i}, \dot{\mathbf{x}}_{i}, \ddot{\hat{\mathbf{x}}}_{i}, \dot{\hat{\mathbf{x}}}_{i})=\operatorname{diag}\{\Upsilon_{i1}(x_{i1}, \dot{x}_{i1}, \ddot{\hat{x}}_{i1}, \dot{\hat{x}}_{i1}),\hspace{-0.08cm}\ldots\hspace{-0.08cm},\hspace{-0.08cm}\Upsilon_{im_i}(x_{im_i}, \dot{x}_{im_i}, \ddot{\hat{x}}_{im_i}, \dot{\hat{x}}_{im_i})\}$, and  $C_{i}(\mathbf{x}_i,\dot{\mathbf{x}}_{i})=\operatorname{diag}\{C_{i1}({x}_{i1},\dot{x}_{i1}),\ldots,C_{im_i}({x}_{im_i},\dot{x}_{im_i})\}$.

The vector form of the auxiliary system \eqref{AuxSys} for the $i$-th coalition is given by
	\begin{align}
		\dot{\eta}_i=&\vartheta_i,\notag \\
		\dot{\vartheta}_i=&\hspace{-0.1cm}-\hspace{-0.1cm}\alpha\vartheta_i\hspace{-0.1cm}-\hspace{-0.1cm}(R_i\xi_i\hspace{-0.1cm}+\hspace{-0.1cm}G_i^{\top}\mathcal{P}_{+}(\lambda_{i})\hspace{-0.1cm}+\hspace{-0.1cm}B_i^{\top}\mathcal{P}_{+}(\omega_i\hspace{-0.1cm}+\hspace{-0.1cm}B_i\eta_i\hspace{-0.05cm}-\hspace{-0.05cm}b_i)),\label{AuxSys_i}
	\end{align}
where $R_i=\operatorname{diag}\{R_{i1},\ldots,R_{im_i}\}$, $\xi_{i}=\operatorname{col}\{\xi_{i1},\ldots,\xi_{im_i}\}$, $\chi_i=\operatorname{col}\{\chi_{i1},\ldots,\chi_{im_i}\}$, $\lambda_i = \operatorname{col}\{\lambda_{i1},  \ldots, \lambda_{im_i}\}$, $G_i=\operatorname{diag}\{G_{i1},\ldots,G_{im_i}\}$, $B_i=\operatorname{diag}\{B_{i1},\ldots,B_{im_i}\}$, $\omega_i = \operatorname{col}\{\omega_{i1}, \ldots, \omega_{im_i}\}$, and $b_i = \operatorname{col}\{b_{i1}, \ldots, b_{im_i}\}$.

The vector form of the Lagrange multiplier update system \eqref{LagMulti} for the $i$-th coalition is given by
\begin{subequations}\label{LagMulti_i}
	\begin{align}
		\dot{\omega}_i=&B_i\vartheta_i-\omega_i+\mathcal{P}_{+}(\omega_i+B_i\eta_i-b_i),\\
		\dot{\lambda}_{i}=&-\lambda_{i}+\mathcal{P}_{+}(\lambda_i)+G_i\eta_i-g_i+G_i\vartheta_{i}\notag \\
		&-(\mathcal{L}_i\otimes I_p)\mathcal{P}_{+}(\lambda_i)-(\mathcal{L}_i\otimes I_p)\rho_i, \label{LagMulti_ic}\\
		\dot{\rho}_i=&(\mathcal{L}_i\otimes I_p)\mathcal{P}_{+}(\lambda_i),
	\end{align}
\end{subequations}
where $g_i = \operatorname{col}\{g_{i1},\ldots,g_{im_i}\}$ and $\rho_i = \operatorname{col}\{\rho_{i1}, \ldots, \rho_{im_i}\}$.

The vector form of the gradient tracking system \eqref{TrackGrad} for the $i$-th coalition is given by
\begin{subequations}\label{TrackGrad_i}
	\begin{align}
		\dot{\xi}_i=&\hspace{-0.05cm}-\hspace{-0.05cm}\beta\big(\xi_i\hspace{-0.05cm}+\hspace{-0.05cm}(\mathcal{L}_i\hspace{-0.05cm}\otimes\hspace{-0.05cm} I_{rm_i})\xi_i\hspace{-0.05cm}+\hspace{-0.05cm}(\mathcal{L}_i\hspace{-0.05cm}\otimes\hspace{-0.05cm} I_{rm_i})\zeta_i\hspace{-0.05cm}-\hspace{-0.05cm}\Gamma_i(\chi_i)\big), \label{coal_if}\\
		\dot{\zeta}_i=&\beta(\mathcal{L}_i\otimes I_{rm_i})\xi_i,
	\end{align}
\end{subequations}
where $\Gamma_i(\chi_i) = \operatorname{col}\big\{\frac{\partial J_{i1}(\chi_{i1})}{\partial x_{i1}}, \frac{\partial J_{i1}(\chi_{i1})}{\partial x_{i2}},\ldots, \frac{\partial J_{i1}(\chi_{i1})}{\partial x_{im_i}}, $ $\frac{\partial J_{i2}(\chi_{i2})}{\partial x_{i1}}, \ldots, \frac{\partial J_{im_i}(\chi_{im_i})}{\partial x_{im_i}}\big\}$.

For the convenience of equilibrium analysis, define $E(\mathbf{x})=\operatorname{diag}\{E_{1}(\mathbf{x}_{1}),\ldots,E_{N}(\mathbf{x}_{N})\}$, $\Upsilon(\mathbf{x}, \dot{\mathbf{x}}, \ddot{\hat{\mathbf{x}}}, \dot{\hat{\mathbf{x}}})=\operatorname{diag}\{\Upsilon_{1}(x_{1}, \dot{x}_{1}, \ddot{\hat{x}}_{1}, \dot{\hat{x}}_{1}),\hspace{-0.02cm}\ldots,\hspace{-0.08cm}\Upsilon_{N}(x_{N}, \dot{x}_{N}, \ddot{\hat{x}}_{N}, \dot{\hat{x}}_{N})\}$, $e=\operatorname{col}\{e_{1},\ldots,e_{N}\}$, $D(\mathbf{x})=\operatorname{diag}\{D_{1}(\mathbf{x}_{1}),\ldots,D_{N}(\mathbf{x}_{N})\}$, $\Psi=\operatorname{diag}\{\Psi_{1},\ldots,\Psi_{N}\}$, $C(\dot{\mathbf{x}})=\operatorname{diag}\{C_{1}(\dot{x}_{1}),\ldots,C_{N}(\dot{x}_{N})\}$, $\hat{\mu}=\operatorname{col}\{\hat{\mu}_{1},\ldots,\hat{\mu}_{N}\}$, $\varepsilon=\operatorname{diag}\{\varepsilon_{1},\ldots,\varepsilon_{N}\}$, $\hat{d}=\operatorname{col}\{\hat{d}_{1},\ldots,\hat{d}_{N}\}$, $d=\operatorname{col}\{d_{1},\ldots,d_{N}\}$, and $\eta=\operatorname{col}\{\eta_{1},\ldots,\eta_{N}\}$. The uncertain closed-loop EL system \eqref{ELSys_i} for all coalitions takes the following form:
\begin{subequations}\label{HighOrder_Co}
	\begin{align}
		&E(\mathbf{x}) \ddot{\mathbf{x}} + C(\dot{\mathbf{x}}) \dot{\mathbf{x}} + D(\mathbf{x})\dot{\mathbf{x}} \notag\\
		=& -\gamma e +\Upsilon(\mathbf{x}, \dot{\mathbf{x}}, \ddot{\hat{\mathbf{x}}}, \dot{\hat{\mathbf{x}}})\hat{\mu}- \operatorname{sgn}(\varepsilon)\hat{d} +\Psi d, \\
		\dot{\hat{\mu}} =& -\Upsilon^{\top}(\mathbf{x}, \dot{\mathbf{x}}, \ddot{\hat{\mathbf{x}}}, \dot{\hat{\mathbf{x}}}) e, \\
		\dot{\hat{d}} =& \operatorname{sgn}(\varepsilon^{\top})e, \\
		e =& \dot{\mathbf{x}} - \dot{\hat{\mathbf{x}}}, \\
		\dot{\hat{\mathbf{x}}} =& \vartheta - (\mathbf{x} - \eta).
	\end{align}
\end{subequations}

Define $R=\operatorname{diag}\{R_{1}, \ldots, R_{N}\}$, $\xi=\operatorname{col}\{\xi_{1}, \ldots, \xi_{N}\}$, $\chi = \operatorname{col}\{\chi_{1}, \ldots, \chi_{N}\}$, $\lambda = \operatorname{col}\{\lambda_{1},  \ldots, \lambda_{N}\}$, $G = \operatorname{diag}\{G_{1}, \ldots, G_{N}\}$, $\omega = \operatorname{col}\{\omega_{1},  \ldots, \omega_{N}\}$, $B=\operatorname{diag}\{B_1,\ldots,B_N\}$, and $b = \operatorname{col}\{b_{1}, \ldots, b_{N}\}$. The auxiliary system \eqref{AuxSys_i} for all coalitions is given by
	\begin{align}
		\dot{\eta} =& \vartheta, \notag\\
		\dot{\vartheta} =& \hspace{-0.07cm}-\hspace{-0.07cm} \alpha\vartheta \hspace{-0.07cm}-\hspace{-0.07cm} \big(R \xi (\chi ) \hspace{-0.07cm}+\hspace{-0.07cm} G^{\top} \mathcal{P}_{+}(\lambda ) \hspace{-0.07cm}+\hspace{-0.07cm} B^{\top}\mathcal{P}_{+}(\omega \hspace{-0.07cm}+\hspace{-0.07cm} B\eta \hspace{-0.07cm}-\hspace{-0.07cm} b )\big).\label{AuxSys_Co}
	\end{align}

The Lagrange multiplier update system \eqref{LagMulti_i} for all coalitions is represented by
\begin{subequations}\label{LagMulti_Co}
	\begin{align}
		\dot{\omega} =& B\vartheta -\omega + \mathcal{P}_{+}(\omega + B \eta - b) , \\
		\dot{\lambda} =& -\lambda + \mathcal{P}_{+}(\lambda) + G \eta - g - (\mathcal{L} \otimes I_p) \mathcal{P}_{+}(\lambda) \\
		&- (\mathcal{L} \otimes I_p) \rho +G\vartheta, \\
		\dot{\rho} =& (\mathcal{L} \otimes I_p) \mathcal{P}_{+}(\lambda),
	\end{align}
\end{subequations}
where $g = \operatorname{col}\{g_1, \ldots, g_N\}$, $\rho = \operatorname{col}\{\rho_1, \ldots, \rho_N\}$, and $\mathcal{L} = \operatorname{diag}\{\mathcal{L}_1, \ldots, \mathcal{L}_N\}$.

Define $\mathcal{L}_{\xi}
=
\operatorname{diag}
\left\{
\mathcal{L}_1\otimes I_{rm_1},
\ldots,
\mathcal{L}_N\otimes I_{rm_N}
\right\}$. The gradient tracking system \eqref{TrackGrad_i} for all coalitions is specified by
\begin{subequations}\label{TrackGrad_Co}
	\begin{align}
		\dot{\xi} =& -\beta \big( \xi + \mathcal{L}_{\xi} \xi + \mathcal{L}_{\xi} \zeta - \Gamma(\chi) \big), \\
		\dot{\zeta} =& \beta \mathcal{L}_{\xi} \xi,
	\end{align}
\end{subequations}
where $\Gamma(\chi) = \operatorname{col}\left\{\Gamma_1(\chi_1), \ldots, \Gamma_N(\chi_N)\right\}$.

Based on the aforementioned analysis, the following lemma characterizes the relationship between the NE and the equilibrium of the auxiliary system.
\begin{lemma}\label{ch5_equi}
	Suppose that Assumptions
	\ref{ExistAss}-\ref{Ass_WeightLMI} hold. Then $\mathbf{x}^{\star}$ is the NE of coalition game $\Game$ if and only if there exists an equilibrium $(\bar{z}^{*},\eta^{*},\vartheta^{*},\omega^{*}, \lambda^{*},\rho^{*},\xi^{*},\zeta^{*})$
	of systems \eqref{s_Sys}, \eqref{AuxSys_Co}, \eqref{LagMulti_Co}, and \eqref{TrackGrad_Co}
	satisfying $\eta^{*}=\mathbf{x}^{\star}$.
\end{lemma}

\begin{proof}
	For all $i \in \mathbb{N}$, in light of \eqref{AuxSys_i}, \eqref{LagMulti_i}, and \eqref{TrackGrad_i}, one can verify that
	\begin{subequations}
		\begin{align}
			\mathbf{0} =& \vartheta_i^{*}, \label{Equ_a}\\
			\mathbf{0} =& -\alpha\vartheta_i^{*} - \big(R_i \xi_i^{*}(\chi_i^{*}) + G_i^{\top} \mathcal{P}_{+}(\lambda_i^{*}) \notag \\
			&+ B_i^{\top}\mathcal{P}_{+}(\omega_i^{*} + B_i\eta_i^{*} - b_i)\big), \\
			\mathbf{0} =& -\omega_i^{*} + \mathcal{P}_{+}(\omega_i^{*} + B_i\eta_i^{*} - b_i) + B_i \vartheta_i^{*}, \label{Equ_c} \\
			\mathbf{0} =& -\lambda_i^{*} + \mathcal{P}_{+}(\lambda_i^{*}) + G_i \eta_i^{*} - g_i +G_i\vartheta_{i}^{*} \notag\\
			&- (\mathcal{L}_i \otimes I_p) \mathcal{P}_{+}(\lambda_i^{*}) - (\mathcal{L}_i \otimes I_p) \rho_i^{*}, \label{Equ_e} \\
			\mathbf{0} =& (\mathcal{L}_i \otimes I_p) \mathcal{P}_{+}(\lambda_i^{*}), \label{Equ_f} \\
			\mathbf{0} =& \hspace{-0.08cm}- \hspace{-0.08cm} \beta \big( \xi_i^{*} \hspace{-0.08cm} + \hspace{-0.08cm} (\mathcal{L}_i \otimes I_{r m_i}) \xi_i^{*} \hspace{-0.08cm} + \hspace{-0.08cm} (\mathcal{L}_i \otimes I_{r m_i}) \zeta_i^{*} \hspace{-0.08cm} - \hspace{-0.08cm} \Gamma_i^{*}(\chi_i^{*}) \big), \label{Equ_g} \\
			\mathbf{0} =& \beta (\mathcal{L}_i \otimes I_{r m_i}) \xi_i^{*}. \label{Equ_h}
		\end{align}
	\end{subequations}
With \eqref{s_Sys}, one can further get that
	\begin{equation}\label{Equ_1}
		\mathbf{0} = -\kappa \Xi \mathbf{L} \Xi^{\top} \bar{z}^{*} - \Xi(\mathbf{1}_n \otimes \vartheta^{*} ).
	\end{equation}
Since $\vartheta_i^* = \mathbf{0}$, $\forall i \in \mathbb{N}$, it follows from \eqref{Equ_1} and $s = \Xi^{\top} \bar{z} + \mathbf{1}_n \otimes \eta$ that $\bar{z}^* = \mathbf{0}$ and $s^* = \mathbf{1}_n \otimes \eta^*$. 

For the Lagrange multipliers associated with the local constraints, substituting \eqref{Equ_a} into \eqref{Equ_c} yields
	\begin{align}\label{asa_conp1}
		\omega_i^{*} =& \mathcal{P}_{+}(\omega_i^{*} + B_i\eta_i^{*} - b_i).
	\end{align}
For all $j\in\mathbb{V}_i$, it follows from \eqref{asa_conp1} that $\omega_{ij}^{*}=P_{+}(\omega_{ij}^{*}+B_{ij}\eta_{ij}^{*}-b_{ij})$. By the definition of the projection operator, this equality directly implies $\omega_{ij}^{*}\geq\mathbf{0}$, $B_{ij}\eta_{ij}^{*}\leq b_{ij}$, and $\omega_{ij}^{*}\odot(B_{ij}\eta_{ij}^{*}-b_{ij})=\mathbf{0}$. Therefore, $\eta_{ij}^{*}$ satisfies \eqref{constraint_ijc}, and $\omega_{ij}^{*}$ satisfies the complementary slackness condition \eqref{constraint_ijf}.

	Moreover, it is easy to see that \eqref{Equ_f} holds if and only if $\mathcal{P}_{+}(\lambda^{*}_{i1}) = \mathcal{P}_{+}(\lambda^{*}_{i2}) = \ldots = \mathcal{P}_{+}(\lambda^{*}_{im_i})$.  Left-multiplying both sides of \eqref{Equ_e} by $(\mathbf{1}_{m_i} \otimes I_p)^{\top}$ yields
	\begin{equation}\label{asa_con1}
		\sum_{j=1}^{m_i} \lambda_{ij}^{*} - \sum_{j=1}^{m_i} \mathcal{P}_{+}(\lambda_{ij}^{*}) = \sum_{j=1}^{m_i} G_{ij} \eta_{ij}^{*} - \sum_{j=1}^{m_i} g_{ij}.
	\end{equation}
	Since $\mathcal{P}_{+}(\lambda_{ij}^*) \geq \lambda_{ij}^{*}$, $\forall j \in \mathbb{V}_{i}$, it is easy to verify that
	\begin{equation*}
		\sum_{j=1}^{m_i} G_{ij} \eta_{ij}^{*} - \sum_{j=1}^{m_i} g_{ij} \leq \mathbf{0},
	\end{equation*}
	which corresponds to \eqref{constraint_ijb}. For each $k\in\{1,\ldots,p\}$, suppose that
	$[\mathcal{P}_{+}(\lambda_{ij}^{*})]_k>0$. The agreement of the projected multipliers implies $[\lambda_{il}^{*}
	-\mathcal{P}_{+}(\lambda_{il}^{*})]_k=0$, $\forall l\in\mathbb{V}_i$.
	It follows from \eqref{asa_con1} that $	\left[
	\sum_{l=1}^{m_i}
	(G_{il}\eta_{il}^{*}-g_{il})
	\right]_k=0$.
	If $[\mathcal{P}_{+}(\lambda_{ij}^{*})]_k=0$,
	the corresponding complementarity condition holds directly.
	Thus, \eqref{constraint_ije} holds componentwise.

	In light of \eqref{Equ_h}, one can verify that $\xi_{i1k}^* = \xi_{i2k}^* = \ldots = \xi_{im_ik}^*$, $\forall k \in \mathbb{V}_i$. Left-multiplying both sides of \eqref{Equ_g} by $\mathbf{1}_{m_i}^{\top} \otimes I_{m_ir}$ gives
	\begin{equation*}
		\sum_{l=1}^{m_i} \xi^*_{ilk} = m_i \frac{\partial J_i(\eta^*)}{\partial x_{ik}}.
	\end{equation*}
	Hence, it is easy to obtain that $\xi_{ijk}^* = \frac{\partial J_i(\eta^*)}{\partial x_{ik}}$, $\forall j \in \mathbb{V}_i, i \in \mathbb{N}$. Therefore, one can further get
	\begin{equation*}
		\frac{\partial J_i(\eta^*)}{\partial x_{i}} + G_i^{\top}\mathcal{P}_{+}(\lambda_{i}^{*}) + B_i^{\top}\mathcal{P}_{+}(\omega_i^{*} + B_i\eta_i^{*} - b_i) = \mathbf{0},
	\end{equation*}
	which corresponds to \eqref{constraint_ija}. Consequently, the equilibrium $\eta^*$ is the NE $\mathbf{x}^{\star}$. 
	
	Conversely, suppose that $\mathbf{x}^{\star}$ is the NE
	of coalition game $\Game$. Let $\omega_{ij}^{\star}$ and
	$\lambda_{ij}^{\star}$ satisfy the KKT conditions
	\eqref{constraints_ij}.
	Set $\eta^{*}=\mathbf{x}^{\star}$, $\vartheta^{*}=\mathbf{0}$, $\bar{z}^{*}=\mathbf{0}$.
	Then, $s^{*}=\mathbf{1}_n\otimes\mathbf{x}^{\star}$ and
	$\chi_i^{*}=\mathbf{1}_{m_i}\otimes\mathbf{x}^{\star}$.
	
	For each $i\in\mathbb{N}$ and $j\in\mathbb{V}_i$, choose $\omega_{ij}^{*}=\omega_{ij}^{\star}$,
	$\lambda_{ij}^{*}
	=\mathcal{P}_{+}(\lambda_{ij}^{\star})
	+\frac{1}{m_i}\sum_{l=1}^{m_i}
	\left(G_{il}x_{il}^{\star}-g_{il}\right)$.
	By feasibility and complementary slackness, $\mathcal{P}_{+}(\lambda_{ij}^{*})=\mathcal{P}_{+}(\lambda_{ij}^{\star})$, $\mathcal{P}_{+}
		\left(\omega_{ij}^{*}+B_{ij}\eta_{ij}^{*}-b_{ij}\right)=\omega_{ij}^{*}$.
	Moreover, the agreement of the projected KKT multipliers gives $(\mathcal{L}_i\otimes I_p)
	\mathcal{P}_{+}(\lambda_i^{*})=\mathbf{0}$.	The definition of $\lambda_{ij}^{*}$ yields
	\[
	(\mathbf{1}_{m_i}^{\top}\otimes I_p)
	\left(
	-\lambda_i^{*}+\mathcal{P}_{+}(\lambda_i^{*})
	+G_i\eta_i^{*}-g_i
	\right)=\mathbf{0}.
	\]
	Since $\mathcal{G}_i$ is undirected and connected,
	there exists $\rho_i^{*}$ satisfying
	\[
	(\mathcal{L}_i\otimes I_p)\rho_i^{*}
	=-\lambda_i^{*}+\mathcal{P}_{+}(\lambda_i^{*})
	+G_i\eta_i^{*}-g_i.
	\]
	Thus, \eqref{Equ_c}, \eqref{Equ_e}, and \eqref{Equ_f} hold.
	
	Next, choose $\xi_i^{*}
	=\mathbf{1}_{m_i}\otimes
	\nabla_{\mathbf{x}_i}J_i(\mathbf{x}^{\star})$. 
	Then, $(\mathcal{L}_i\otimes I_{rm_i})\xi_i^{*}
	=\mathbf{0}$. By the definition of $J_i$,
	\begin{align*}
		(\mathbf{1}_{m_i}^{\top}\otimes I_{rm_i})
		\left(\Gamma_i(\chi_i^{*})-\xi_i^{*}\right)
		=\mathbf{0}.
	\end{align*}
	Therefore, there exists $\zeta_i^{*}$ such that
	\[
	(\mathcal{L}_i\otimes I_{rm_i})\zeta_i^{*}
	=\Gamma_i(\chi_i^{*})-\xi_i^{*},
	\]
	which establishes \eqref{Equ_g} and \eqref{Equ_h}.
	
	Finally, $R_i\xi_i^{*}
	=\nabla_{\mathbf{x}_i}J_i(\mathbf{x}^{\star})$.
	Together with \eqref{constraint_ija}, this gives $R_i\xi_i^{*}
	+G_i^{\top}\mathcal{P}_{+}(\lambda_i^{*})
	+B_i^{\top}\mathcal{P}_{+}
	\left(\omega_i^{*}+B_i\eta_i^{*}-b_i\right)
	=\mathbf{0}$.
	Since $\vartheta^{*}=\mathbf{0}$ and
	$\bar{z}^{*}=\mathbf{0}$, all the equilibrium equations
	\eqref{Equ_a}-\eqref{Equ_h} and \eqref{Equ_1} are satisfied.
	Hence,
	$(\bar{z}^{*},\eta^{*},\vartheta^{*},\omega^{*},
	\lambda^{*},\rho^{*},\xi^{*},\zeta^{*})$
	is an equilibrium of the auxiliary system with
	$\eta^{*}=\mathbf{x}^{\star}$.
\end{proof}

Define an orthogonal matrix $[\frac{\mathbf{1}_{m_i}}{\sqrt{m_i}}, ~U_i]$ such that $\frac{\mathbf{1}^{\top}_{m_i}}{\sqrt{m_i}}U_i = \mathbf{0}$, $U_i^{\top}U_i = I_{m_i-1}$, and $U_iU_i^{\top} = I_{m_i} - \frac{1}{m_i}\mathbf{1}_{m_i}\mathbf{1}_{m_i}^{\top}$. Decompose $\zeta_i$ as $\operatorname{col}\{\zeta_i^1, \zeta_i^2\} = ([\frac{\mathbf{1}_{m_i}}{\sqrt{m_i}}, ~U_i]^{\top} \otimes I_{rm_i}) \zeta_i$. Thus, it is easy to see that
\begin{subequations}
	\begin{align}
		\dot{\zeta}_i^1 =& \mathbf{0}, \\
		\dot{\zeta}_i^2 =& \beta(U_{i}^{\top}\mathcal{L}_i \otimes I_{rm_i}) \xi_i. \label{zeta_i_2}
	\end{align}
\end{subequations}
In light of Assumption \ref{ch5_AssGraph}, the matrix
\[
M_i = \left[ {\begin{array}{*{20}{c}}
		{-I_{rm_i^2} - \mathcal{L}_i \otimes I_{rm_i}} & -\mathcal{L}_i U_i \otimes I_{rm_i} \\
		{U_{i}^{\top} \mathcal{L}_i \otimes I_{rm_i}} & \mathbf{0}
\end{array}} \right]
\]
is Hurwitz \cite{moreau2004}. Therefore, there exist positive definite matrices $P_i$ and $Q_i$ such that
\[
P_i M_i + M_i^{\top} P_i = -Q_i.
\]

Let $v_i = \operatorname{col}\{ \xi_i, \zeta_i^{2} \}$. By the definition of $M_i$, it follows from \eqref{coal_if} and  \eqref{zeta_i_2} that
\begin{align}
	\dot{v}_i =& \beta M_i v_i + \beta \left[ {\begin{array}{*{20}{c}}
			{\Gamma_i(\chi_i)} \\
			\mathbf{0}
	\end{array}} \right].
\end{align}
Furthermore, let $\xi_i^{q}$ and $\zeta_i^{2q}$ denote the quasi-steady states of the gradient tracking system \eqref{coal_if} and \eqref{zeta_i_2}, respectively, which satisfy
\[
\xi_i^q = \frac{\mathbf{1}_{m_i} \mathbf{1}_{m_i}^{\top} \otimes I_{rm_i}}{m_i} \Gamma_i(\chi_i),
\]
and
\begin{equation*}
	\left[ {\begin{array}{*{20}{c}}
			{-I_{rm_i^2} - \mathcal{L}_i \otimes I_{rm_i}} & -\mathcal{L}_i U_i \otimes I_{rm_i} \\
			{U_i^{\top} \mathcal{L}_i \otimes I_{rm_i}} & \mathbf{0}
	\end{array}} \right] \hspace{-0.1cm} v_i^q + \left[ {\begin{array}{*{20}{c}}
			{\Gamma_i(\chi_i)} \\
			\mathbf{0}
	\end{array}} \right] \hspace{-0.1cm}= \mathbf{0}
\end{equation*}
with $v_i^q = \operatorname{col}\{ \xi_i^q, \zeta_i^{2q} \}$.

Define $\bar{v}_i^q = v_i - v_i^{q}$, where $v_i^{q}$ satisfies $M_i v_i^{q} + H_i(\chi_i) = \mathbf{0}$ with $H_i(\chi_i) = \operatorname{col}\{ \Gamma_i(\chi_i),\mathbf{0} \}$. Therefore, it is easy to see that
\begin{equation}\label{ch5_dot_v_i_q}
	\dot{\bar{v}}_i^q = \beta M_i \bar{v}_i^q - \dot{v}_i^q.
\end{equation}

\begin{theorem}\label{lemma_axuEqu}
	Suppose Assumptions \ref{ExistAss}--\ref{Ass_WeightLMI} hold. Let the gains $\alpha$, $\beta$, $\gamma$, and $\kappa$ satisfy $\gamma>1/4$ and
	\begin{equation*}
		\begin{aligned}
			\alpha &> 2+\|B\|^2 + \frac{2N\ell^2}{\hbar} + \frac{n}{2}
			+\bar{\ell}^2c_p^2(n+1)^2,\\
			\kappa &> \frac{2}{\lambda_{\mathbf{L}}}
			\left(\frac{3}{2}+\frac{\ell^2}{\hbar}+\frac{\ell^2}{2}\right),\\
			\beta &> \frac{1}{c_q}
			\left(\bar{\ell}^2c_p^2\kappa^2\sigma_{\max}^2(\mathbf{L})
			+\frac{1}{\hbar}+\frac{3}{2}\right),
		\end{aligned}
	\end{equation*}
	where $\lambda_{\mathbf{L}}=\lambda_{\min}(\Xi(\mathbf{L}+\mathbf{L}^{\top})\Xi^{\top})$,
	$\sigma_{\max}^2(\mathbf{L})=\lambda_{\max}(\mathbf{L}^{\top}\mathbf{L})$,
	$\bar{\ell}=\ell\max_{i\in\mathbb{N}}\|M_i^{-1}\|$, and $c_p,c_q>0$ are constants to be defined later. Then, under the distributed strategy
	\eqref{USV_controller}, \eqref{AuxSys}, \eqref{LagMulti},
	\eqref{TrackGrad}, and \eqref{EstAct},
	the physical states of all USVs converge asymptotically
	to the unique NE $\mathbf{x}^{\star}$.
\end{theorem}

\begin{proof} 
	The proof proceeds in two steps.		
	
	\textit{Step 1}: Due to the discontinuity introduced by the sign function, the closed-loop system is understood in the Filippov sense throughout the proof. To facilitate the subsequent analysis, define $\bar{\eta}_i = \eta_i - \eta_i^*$, $\bar{\vartheta}_i = \vartheta_i - \vartheta_i^*$, $\bar{\xi}_i = \xi_i - \xi_i^*$, $\bar{\omega}_i = \omega_i - \omega_i^*$, and $\bar{\rho}_i = \rho_i - \rho_i^*$. In light of Lemma \ref{ch5_equi}, one can verify that
	\begin{subequations}
		\begin{align}
			\dot{\bar{\eta}}_i =& \bar{\vartheta}_i, \label{error_a}\\
			\dot{\bar{\vartheta}}_i =& -\alpha \bar{\vartheta}_i - R_i \bar{\xi}_i - G_i^{\top} \mathcal{P}_{+}(\lambda_i) + G_i^{\top} \mathcal{P}_{+}(\lambda_i^*) \notag\\
			&- B_i^{\top}\mathcal{P}_{+}(\omega_i + B_i\eta_i - b_i) + B_i^{\top}\mathcal{P}_{+}(\omega_i^* + B_i\eta_i^* - b_i), \label{error_b}\\
			\dot{\bar{\omega}}_i =& -\omega_i + \mathcal{P}_{+}(\omega_i +B_i\eta_i - b_i) + B_i\vartheta_i, \label{error_d}\\
			\dot{\bar{\rho}}_i =& (\mathcal{L}_i \otimes I_p) (\mathcal{P}_{+}(\lambda_i) - \mathcal{P}_{+}(\lambda_i^*)). \label{error_f}
		\end{align}
	\end{subequations}

	In light of Lemma \ref{Projection_property}, one can get $\|\lambda_i - \mathcal{P}_{+}(\lambda_i^*)\|^2 - \|\lambda_i-\mathcal{P}_{+}(\lambda_i)\|^2 \geq 0$. Therefore, consider the following Lyapunov function candidate
	\begin{align*}
		V = \frac{1}{2}  \|\bar{z}\|^2 + \sum_{i=1}^{N} V_i,
	\end{align*}
	where $V_i = V_{i1} + V_{i2}  + V_{i3}$, and
	\begin{align*}
		V_{i1} =& \frac{1}{2} (\alpha - 1) \|\bar{\eta}_i\|^2 + \frac{1}{2} \|\bar{\eta}_i + \bar{\vartheta}_i\|^2+\frac{1}{2} \|\bar{\omega}_i\|^2, \\
		V_{i2} =& \frac{1}{2} \|\lambda_i - \mathcal{P}_{+}(\lambda_i^*)\|^2 - \frac{1}{2} \|\lambda_i - \mathcal{P}_{+}(\lambda_i)\|^2 + \frac{1}{2} \|\bar{\rho}_i\|^2,\\
		V_{i3}=&(\bar{v}_i^q)^{\top} P_i \bar{v}_i^q.
	\end{align*}
	
	The time derivative of $V_{i1}$ along the trajectories of \eqref{error_a} and \eqref{error_b} is given by
	\begin{equation}\label{ch5_Vi1}
		\begin{aligned}
			\dot{V}_{i1} =& -(\alpha-1)\|\bar{\vartheta}_i\|^2 - \bar{\eta}_i^{\top}R_i\bar{\xi}_i - \bar{\vartheta}_i^{\top}R_i\bar{\xi}_i-\|\bar{\omega}_i\|^2	\\
			&- \bar{\eta}_i^{\top}B_i^{\top}(\mathcal{P}_{+}(\omega_i + B_i\eta_i - b_i) - \mathcal{P}_{+}(\omega_i^{*} + B_i\eta_i^{*} - b_i)) \\
			& - \bar{\vartheta}_i^{\top}B_i^{\top}(\mathcal{P}_{+}(\omega_i + B_i\eta_i - b_i) - \mathcal{P}_{+}(\omega_i^{*} + B_i\eta_i^{*} - b_i)) \\
			& - (\bar{\eta}_i+\bar{\vartheta}_i)^{\top}G_i^{\top}(\mathcal{P}_{+}(\lambda_i) - \mathcal{P}_{+}(\lambda_i^{*}))+\bar{\omega}_i^{\top}B_i\bar{\vartheta}_{i}\\
			&+\bar{\omega}_i^{\top}(\mathcal{P}_{+}(\omega_i + B_i\eta_i - b_i) - \mathcal{P}_{+}(\omega_i^{*} + B_i\eta_i^{*} - b_i)).
		\end{aligned}
	\end{equation}
Furthermore, with \eqref{asa_conp1}, one can verify that
	\begin{equation}\label{Vi1_Ext2}
		\begin{aligned}
			&-\|\omega_i-\omega_i^{*}\|^2+(\omega_i-\omega_i^{*})^{\top}\\
			&\times\big(\mathcal{P}_{+}(\omega_i + B_i\eta_i - b_i) - \mathcal{P}_{+}(\omega_i^{*} + B_i\eta_i^{*} - b_i)\big)\\
			=&-\|\omega_{i}-\mathcal{P}_{+}(\omega_i + B_i\eta_i - b_i)\|^2\\
			&-(\mathcal{P}_{+}(\omega_i + B_i\eta_i - b_i)-\omega_i^{*})^{\top}(\omega_i-\mathcal{P}_{+}(\omega_i + B_i\eta_i - b_i)).
		\end{aligned}
	\end{equation}
In light of Lemma \ref{Projection_property}, it is easy to see that
	\begin{equation}\label{Proj_Local_Lagr}
		\begin{aligned}
			&\left(\omega_{i}+B_{i}\eta_{i}-b_{i}-\mathcal{P}_{+}(\omega_{i}+B_{i}\eta_{i}-b_{i})\right)^{\top}\\
			&\times\left(\mathcal{P}_{+}(\omega_{i}+B_{i}\eta_{i}-b_{i})-\mathcal{P}_{+}(\omega_{i}^{*}+B_{i}\eta_{i}^{*}-b_{i})\right)\geq 0.
		\end{aligned}
	\end{equation}
Furthermore, it follows from \eqref{Vi1_Ext2} and \eqref{Proj_Local_Lagr} that
	\begin{equation}\label{Vi1_Ext1}
		\begin{aligned}
			&-\|\omega_i-\omega_i^{*}\|^2-\left(\eta_i-\eta_i^{*}\right)^{\top}B_i^{\top}\\
			&\times \left(\mathcal{P}_{+}(\omega_i + B_i\eta_i - b_i) - \mathcal{P}_{+}(\omega_i^{*} + B_i\eta_i^{*} - b_i)\right)\\
			&+(\omega_i-\omega_i^{*})^{\top}\left(\mathcal{P}_{+}(\omega_i + B_i\eta_i - b_i) - \mathcal{P}_{+}(\omega_i^{*} + B_i\eta_i^{*} - b_i)\right)\\
			\leq&-\|\omega_{i}-\mathcal{P}_{+}(\omega_i + B_i\eta_i - b_i)\|^2-(\omega_i^{*})^{\top}(B_i\eta_i^{*}-b_i)\\
			&+(\mathcal{P}_{+}(\omega_i + B_i\eta_i - b_i))^{\top}(B_i\eta_i^{*}-b_i)\\
			\leq&-\|\omega_{i}-\mathcal{P}_{+}(\omega_i + B_i\eta_i - b_i)\|^2,
		\end{aligned}
	\end{equation}
	where the last inequality holds by $\mathcal{P}_{+}(\omega_i + B_i\eta_i - b_i)\geq \mathbf{0}$, $B_i\eta_i^{*}\leq b_i$, and $(\omega_i^{*})^{\top}(B_i\eta_i^{*}-b_i)=0$.
	
	By Young's inequality, one can further get
		\begin{align}
		&	\bar{\omega}_i^{\top} B_i\bar{\vartheta}_i - \bar{\vartheta}_i^{\top}B_i^{\top} \notag \\
		&\times \left( \mathcal{P}_{+}(\omega_i + B_i\eta_i - b_i) - \mathcal{P}_{+}(\omega_i^* + B_i\eta_i^* - b_i) \right) \notag \\
			\leq & \|B_i\|^2 \|\bar{\vartheta}_i\|^2 + \frac{1}{4} \|\omega_i - \mathcal{P}_{+}(\omega_i + B_i\eta_i - b_i)\|^2.\label{ch5_Vi5_young}
		\end{align}
Hence, substituting \eqref{Vi1_Ext1} and \eqref{ch5_Vi5_young} into \eqref{ch5_Vi1} yields
\begin{equation}\label{ch5_Vi1_New}
	\begin{aligned}
		\dot{V}_{i1} \leq&
		-(\alpha-1-\|B_i\|^2)\|\bar{\vartheta}_i\|^2
		-\bar{\eta}_i^{\top}R_i\bar{\xi}_i
		-\bar{\vartheta}_i^{\top}R_i\bar{\xi}_i\\
		&-\frac{3}{4}\|\omega_{i}-\mathcal{P}_{+}(\omega_i+B_i\eta_i-b_i)\|^2\\
		&-(\bar{\eta}_i+\bar{\vartheta}_i)^{\top}G_i^{\top}
		\big(\mathcal{P}_{+}(\lambda_i)-\mathcal{P}_{+}(\lambda_i^{*})\big).
	\end{aligned}
\end{equation}
Moreover, in light of Lemma \ref{Projection_property}, the time derivative of $V_{i2}$ along the trajectories of \eqref{LagMulti_ic}, \eqref{Equ_e}, and \eqref{error_f} can be obtained as
\begin{align}
	\dot{V}_{i2}
	\leq&
	\left(
	\mathcal{P}_{+}(\lambda_i)
	-\mathcal{P}_{+}(\lambda_i^{*})
	\right)^{\top}
	G_i\bar{\eta}_i
	\notag\\
	&+
	\left(
	\mathcal{P}_{+}(\lambda_i)
	-\mathcal{P}_{+}(\lambda_i^{*})
	\right)^{\top}
	G_i\bar{\vartheta}_i
	\notag\\
	&-
	\left(
	\mathcal{P}_{+}(\lambda_i)
	-\mathcal{P}_{+}(\lambda_i^{*})
	\right)^{\top}
	(\mathcal{L}_i\otimes I_p)
	\left(
	\mathcal{P}_{+}(\lambda_i)
	-\mathcal{P}_{+}(\lambda_i^{*})
	\right).
	\label{ch5_Vi2}
\end{align}
The time derivative of $V_{i3}$ along the trajectory of \eqref{ch5_dot_v_i_q} is given by
		\begin{align}
			\dot{V}_{i3}= & (\bar{v}^q_i)^{\top} P_i \left( \beta M_i \bar{v}_i^q - \dot{v}_i^q \right)+ \left( \beta M_i \bar{v}_i^q - \dot{v}_i^q \right)^{\top} P_i \bar{v}^q_i \notag \\
			\leq & - \hspace{-0.08cm} \beta \lambda_{\min}(Q_i) \|\bar{v}^q_i\|^2\hspace{-0.08cm}-\hspace{-0.08cm}2 (\bar{v}_i^q)^{\top} P_i  \dot{v}_i^q . \label{ch5_Vi7}
		\end{align}
By combining \eqref{ch5_Vi1_New}, \eqref{ch5_Vi2}, and \eqref{ch5_Vi7}, one has
	\begin{align*}
		\dot{V}_i \leq & -(\alpha-1-\|B_i\|^2)\|\bar{\vartheta}_i\|^2 - \bar{\eta}_i^{\top} R_i \bar{\xi}_i - \bar{\vartheta}_i^{\top} R_i \bar{\xi}_i \\
		&- (\mathcal{P}_{+}(\lambda_{i}) - \mathcal{P}_{+}(\lambda_{i}^{*}))^{\top} (\mathcal{L}_i \otimes I_p)(\mathcal{P}_{+}(\lambda_{i}) - \mathcal{P}_{+}(\lambda_{i}^{*})) \\
		& - \frac{3}{4} \|\omega_{i} - \mathcal{P}_{+}(\omega_{i} + B_i\eta_{i} - b_{i})\|^2 - \beta \lambda_{\min}(Q_i) \|\bar{v}_i^q\|^2 \\
		&-2 (\bar{v}_i^q)^{\top} P_i \dot{v}_i^q .
	\end{align*}
Then, one can further get
	\begin{equation}\label{ch5_vvv1}
		\begin{aligned}
			\sum_{i=1}^{N} \dot{V}_i
			\leq & -(\alpha-1-\|B\|^2)\|\bar{\vartheta}\|^2 -  \beta c_q \|\bar{v}^q\|^2 - 2 (\bar{v}^q)^{\top} P \dot{v}^q   \\
			&- (\mathcal{P}_{+}(\lambda) - \mathcal{P}_{+}(\lambda^{*}))^{\top} (\mathcal{L} \otimes I_p) (\mathcal{P}_{+}(\lambda) - \mathcal{P}_{+}(\lambda^{*})) \\
			& - \frac{3}{4} \|\omega - \mathcal{P}_{+}(\omega + B\eta - b)\|^2 - \bar{\eta}^{\top} R \bar{\xi} - \bar{\vartheta}^{\top} R \bar{\xi},
		\end{aligned}
	\end{equation}
	where $c_q=\min_{i \in \mathbb{N}} \{\lambda_{\min}(Q_i)\}$, $\bar{v}^q = \operatorname{col}\{\bar{v}_1^q, \ldots, \bar{v}_N^q\}$, ${v}^q = \operatorname{col}\{ {v}_1^q, \ldots, {v}_N^q\}$ and $P = \operatorname{diag}\{P_1, \ldots, P_N\}$.
	
	The time derivative of $\frac{1}{2}\|\bar{z}\|^2$ along the trajectory of \eqref{s_Sys} is obtained as
	\begin{equation}\label{Ext_V_Z}
			\frac{1}{2} \frac{d\|\bar{z}\|^2}{dt}	\leq  \ -\frac{\kappa}{2} \lambda_{\mathbf{L}} \|\bar{z}\|^2 - \bar{z}^{\top} \Xi (\mathbf{1}_n \otimes \vartheta ).
	\end{equation}

	By Theorem \ref{Thm_WeightNE}, it is easy to verify that
	\begin{equation}\label{ch5_vvv2}
		\begin{aligned}
			-\bar{\eta}^{\top} R \bar{\xi} = & -\bar{\eta}^{\top} \big( R\xi(s) - R\xi^q(s) + R\xi^q(s) - F(\eta) \\
			&+ F(\eta) - F(\eta^*) \big) \\
			\leq & \ \|\bar{\eta}\| \|\bar{v}^q\| + \ell \|\bar{\eta}\| \|\bar{z}\| - \hbar \|\bar{\eta}\|^2 \\
			\leq & \ \frac{1}{\hbar} \|\bar{v}^q\|^2 + \frac{\ell^2}{\hbar} \|\bar{z}\|^2 - \frac{\hbar}{2} \|\bar{\eta}\|^2,
		\end{aligned}
	\end{equation}
	where the last inequality holds by $\|\bar{\eta}\|\|\bar{v}^q\| \leq \frac{\hbar}{4} \|\bar{\eta}\|^2 + \frac{1}{\hbar} \|\bar{v}^q\|^2$ and $\ell \|\bar{\eta}\|\|\bar{z}\| \leq \frac{\hbar}{4} \|\bar{\eta}\|^2 + \frac{\ell^2}{\hbar} \|\bar{z}\|^2$.

	Similarly to \eqref{ch5_vvv2}, it holds that
		\begin{align}
			-\bar{\vartheta}^{\top} R \bar{\xi} = & -\bar{\vartheta}^{\top} \big( R\xi(s) - R\xi^q(s) + R\xi^q(s) - F(\eta) \notag \\
			&+ F(\eta) - F(\eta^*) \big) \notag \\
			\leq & \ \|\bar{\vartheta}\| \|\bar{v}^q\| + \ell \|\bar{\vartheta}\| \|\bar{z}\| + \sqrt{N}\ell \|\bar{\vartheta}\| \|\bar{\eta}\| \notag \\
			\leq & \ (1 \hspace{-0.1cm}+ \hspace{-0.1cm} \frac{2N\ell^2}{\hbar}) \|\bar{\vartheta}\|^2  \hspace{-0.1cm}+ \hspace{-0.1cm} \frac{1}{2} \|\bar{v}^q\|^2  \hspace{-0.1cm}+  \hspace{-0.1cm} \frac{\ell^2}{2} \|\bar{z}\|^2  \hspace{-0.1cm}+  \hspace{-0.1cm} \frac{\hbar}{8} \|\bar{\eta}\|^2,\label{ch5_vvv3}
		\end{align}
	where the last inequality holds by $\|\bar{\vartheta}\|\|\bar{v}^q\|\hspace{-0.1cm} \leq \hspace{-0.1cm}\frac{1}{2} \|\bar{\vartheta}\|^2 \hspace{-0.1cm}+\hspace{-0.1cm} \frac{1}{2} \|\bar{v}^q\|^2$, $\sqrt{N}\ell\|\bar{\vartheta}\|\|\bar{\eta}\|\hspace{-0.1cm}\leq\hspace{-0.1cm}\frac{\hbar}{8}\|\bar{\eta}\|^2 \hspace{-0.1cm}+ \hspace{-0.1cm} \frac{2N\ell^2}{\hbar}\|\bar{\vartheta}\|^2$, and $\ell \|\bar{\vartheta}\|\|\bar{z}\| \hspace{-0.1cm} \leq \hspace{-0.1cm} \frac{1}{2} \|\bar{\vartheta}\|^2 \hspace{-0.1cm} + \hspace{-0.1cm} \frac{\ell^2}{2} \|\bar{z}\|^2$.

	By the definition of \(v_i^q\), the mapping \(v_i^q(\chi_i)\) is globally Lipschitz with constant \(\ell\|M_i^{-1}\|\). Consequently,	$\|\dot{v}_i^q(t)\|	\leq \bar{\ell}\|\dot{\chi}_i(t)\|$	holds for almost every $t\geq0$. Therefore, one can get
	\begin{equation}\label{ch5_vvv4}
	\begin{aligned}
		&-2(\bar{v}^q)^\top P\dot{v}^q(\chi)\\
		\leq & 2\bar{\ell}\|\bar{v}^q\|\|P\|\|\dot{\chi}\|\\
		\leq & 2\bar{\ell}\kappa c_p\sigma_{\max}(\mathbf{L})
		\|\bar{v}^q\|\|\bar{z}\|
		+2\bar{\ell}c_p(n+1)
		\|\bar{v}^q\|\|\bar{\vartheta}\|\\
		\leq &
		(\bar{\ell}^2\kappa^2c_p^2
		\sigma_{\max}^2(\mathbf{L})+1)
		\|\bar{v}^q\|^2
		\hspace{-0.05cm}+\hspace{-0.05cm}\|\bar{z}\|^2
		\hspace{-0.05cm}+\hspace{-0.05cm}\bar{\ell}^2c_p^2(n+1)^2
		\|\bar{\vartheta}\|^2,
	\end{aligned}
\end{equation}
	where ${c}_{p} = \max_{i \in \mathbb{N}} \{ \lambda_{\max}(P_i) \}$.

	In light of Lemma \ref{ch5_equi}, it follows from \eqref{Equ_a} that
	\begin{equation}\label{vvv1}
		-\bar{z}^{\top} \Xi (\mathbf{1}_n \otimes \vartheta) \leq \frac{1}{2}\|\bar{z}\|^2 +  \frac{n}{2} \|\bar{\vartheta}\|^2.
	\end{equation}

	By substituting \eqref{ch5_vvv2}, \eqref{ch5_vvv3}, and \eqref{ch5_vvv4} into \eqref{ch5_vvv1}, followed by substituting \eqref{vvv1} into \eqref{Ext_V_Z}, one can obtain
	\begin{equation*}
		\begin{aligned}
			\dot{V}\leq& -\left(\alpha-2-\|B\|^2-\frac{n}{2}-\frac{2N\ell^2}{\hbar}-\bar{\ell}^2c_p^2(n+1)^2\right)\|\bar{\vartheta}\|^2\\
			&-\hspace{-0.1cm} \left(\beta c_q\hspace{-0.05cm}-\hspace{-0.05cm}\bar{\ell}^2 {c}^2_{p}\kappa^2\sigma_{\max}^2(\mathbf{L})\hspace{-0.05cm}-\hspace{-0.05cm}\frac{1}{\hbar}\hspace{-0.05cm}-\hspace{-0.05cm}\frac{3}{2}\right)\|\bar{v}^q\|^2\\
			&-\left(\frac{\kappa}{2}\lambda_{\mathbf{L}} -\frac{3}{2}-\frac{\ell^2}{\hbar} -\frac{\ell^2}{2}\right)\|\bar{z}\|^2\\
			&-\frac{3}{4}\|\omega-\mathcal{P}_{+}(\omega+B\eta-b)\|^2-\frac{3\hbar}{8}\|\bar{\eta}\|^2\\
			&-(\mathcal{P}_{+}(\lambda)-\mathcal{P}_{+}(\lambda^{*}))^{\top}(\mathcal{L}\otimes I_p)(\mathcal{P}_{+}(\lambda)-\mathcal{P}_{+}(\lambda^{*})).
		\end{aligned}
	\end{equation*}
Since $\mathcal{L}_i$ is a positive semi-definite matrix for all $i \in \mathbb{N}$, select $\alpha$, $\beta$, and $\kappa$ such that 
	\begin{equation*}
		\begin{aligned}
			\alpha > &2+\|B\|^2 + \frac{2N\ell^2}{\hbar} + \frac{n}{2}+\bar{\ell}^2c_p^2(n+1)^2,\\
			\kappa > &\frac{2}{\lambda_{\mathbf{L}}} \left( \frac{3}{2}  + \frac{\ell^2}{\hbar} + \frac{\ell^2}{2} \right), \\
			\beta > &\frac{1}{c_q} \left( \bar{\ell}^2 {c}_{p}^2\kappa^2\sigma_{\max}^2(\mathbf{L}) + \frac{1}{\hbar} +\frac{3}{2} \right),
		\end{aligned}
	\end{equation*}
	which implies 
	\begin{equation}\label{V_Ext_End}
		\dot{V}(t) \leq - \psi^{\top}(t)\Lambda\psi(t)
	\end{equation}
	with $\psi = \operatorname{col}\{\bar{\eta}, \bar{\vartheta}, \bar{v}^q, \bar{z}\}$, $c_{\Lambda_2}=\alpha-2-\|B\|^2-\frac{n}{2}-\frac{2N\ell^2}{\hbar}-\bar{\ell}^2c_p^2(n+1)^2$, $c_{\Lambda_3}=\beta c_q-\bar{\ell}^2 {c}^2_{p}\kappa^2\sigma_{\max}^2(\mathbf{L})-\frac{1}{\hbar}-\frac{3}{2}$, and
	\begin{equation*}
		\Lambda = \left[ {\begin{array}{*{20}{c}}
				{\frac{3\hbar}{8}}I & 0 & 0 & 0 \\
				0 & c_{\Lambda_2}I & 0 & 0 \\
				0 & 0 & c_{\Lambda_3}I & 0 \\
				0 & 0 & 0 & (\frac{\kappa}{2}\lambda_{\mathbf{L}} -\frac{3}{2}-\frac{\ell^2}{\hbar} -\frac{\ell^2}{2})I
		\end{array}} \right] \succ \mathbf{0}.
	\end{equation*}
	
	Consequently, $\bar{\eta}$, $\bar{\vartheta}$, $\bar{\omega}$,
	$\bar{\rho}$, $\bar{v}_i^{q}$, and $\bar{z}$ are all bounded.
	Moreover, by Lemma \ref{Projection_property} and the definition of $V$,
	$P_+(\lambda_i)$ and $\rho_i$ are bounded. Define the positive and
	negative parts of $\lambda_i$ as
	$\lambda_i^{+}=P_+(\lambda_i)$ and
	$\lambda_i^{-}=\max\{-\lambda_i,0\}$, respectively.
	Together with the boundedness of $\eta_i$ and $\vartheta_i$, all terms in
	\[
	h_i=G_i\eta_i-g_i+G_i\vartheta_i
	-(L_i\otimes I_p)P_+(\lambda_i)
	-(L_i\otimes I_p)\rho_i
	\]
	are bounded. Hence, there exists a constant $c_i^h>0$ such that
	$\|h_i\|\leq c_i^h$. From the dynamics of $\lambda_i$, each component
	of $\lambda_i^{-}$ satisfies, almost everywhere, $\dot{\lambda}_{ik}^{-}\leq-\lambda_{ik}^{-}+c_i^h$.
	By the comparison principle, $\lambda_i^{-}$ is bounded. Since
	$\lambda_i^{+}=P_+(\lambda_i)$ is also bounded and
	$\lambda_i=\lambda_i^{+}-\lambda_i^{-}$, $\lambda_i$ is bounded.
	
	Using \eqref{error_a} and the boundedness of $\bar{\vartheta}$,
	$\bar{\eta}$ is uniformly continuous. Furthermore, since
	$\dot V(t)\leq0$ and $V(t)\geq0$, $V(t)$ is non-increasing and bounded
	from below, and hence $V(+\infty)$ exists and is finite. Therefore,
	integrating \eqref{V_Ext_End} over $[0,+\infty)$ yields
	\begin{equation*}
		\int_{0}^{+\infty}\frac{3\hbar}{8}\|\bar{\eta}(t)\|^2dt
		\leq V(0)-V(+\infty).
	\end{equation*}
	By Barbalat's Lemma \cite{khalilNonlinearSystems2002},
	$\bar{\eta}(t)\to0$, and hence $\eta(t)$ asymptotically converges to
	the NE $\mathbf{x}^{\star}$.

	\textit{Step 2}: For all $i\in\mathbb{N}$ and $j\in\mathbb{V}_i$, substituting the controller \eqref{USV_controller_a} into the EL system \eqref{USV_EL_Sys} yields
		\begin{align}
			&E_{ij}(x_{ij}) \ddot{x}_{ij} + C_{ij}({x}_{ij},\dot{x}_{ij}) \dot{x}_{ij} + D_{ij}(x_{ij})  \dot{x}_{ij} \notag \\
			=&\Upsilon_{ij}(x_{ij},\dot{x}_{ij},\ddot{\hat{x}}_{ij},\dot{\hat{x}}_{ij})\hat{\mu}_{ij}-\gamma e_{ij}-\hat{d}_{ij}\operatorname{sgn}(e_{ij})+\Psi_{ij}d_{ij}.\label{ELPf1}
		\end{align}
By the linear parameterization in \eqref{USV_LinearParameterization}, one has
	\begin{equation}\label{ELPf2}
		\begin{aligned}
		&E_{ij}(x_{ij})\ddot{\hat{x}}_{ij}+C_{ij}({x}_{ij},\dot{x}_{ij})\dot{\hat{x}}_{ij}+D_{ij}(x_{ij}) \dot{\hat{x}}_{ij}\\
		=&\Upsilon_{ij}(x_{ij},\dot{x}_{ij},\ddot{\hat{x}}_{ij},\dot{\hat{x}}_{ij}){\mu}_{ij}.
		\end{aligned}		
	\end{equation}
Subtracting \eqref{ELPf2} from \eqref{ELPf1} yields
		\begin{align}
			E_{ij}(x_{ij})\dot{e}_{ij}=&-C_{ij}({x}_{ij},\dot{x}_{ij})e_{ij}-D_{ij}({x}_{ij})e_{ij}-\gamma e_{ij}\notag \\
			&+\Upsilon_{ij}(x_{ij},\dot{x}_{ij},\ddot{\hat{x}}_{ij},\dot{\hat{x}}_{ij})(\hat{\mu}_{ij}-{\mu}_{ij}) \notag \\
			&-\hat{d}_{ij}\operatorname{sgn}(e_{ij})+\Psi_{ij}d_{ij}.\label{ELControl}
		\end{align}

	Consider the following Lyapunov function candidate
\begin{align*}
	W_{ij} =&\frac{1}{2}\|x_{ij}-\eta_{ij}\|^2+ \frac{1}{2}e_{ij}^{\top}E_{ij}(x_{ij})e_{ij} \\
	&+ \frac{1}{2}\|{\mu}_{ij}-\hat{\mu}_{ij}\|^2 + \frac{1}{2}\|\tilde{d}_{ij}-\hat{d}_{ij}\|^2.
\end{align*}

	Along any Filippov solution, the time derivative of $W_{ij}$ exists almost everywhere and satisfies
	\begin{equation*}
		\begin{aligned}
			\dot{W}_{ij} 
			\leq &-\|x_{ij}-\eta_{ij}\|^2+\|e_{ij}\|\|x_{ij}-\eta_{ij}\|\\
			&-\gamma e_{ij}^{\top}e_{ij} + e_{ij}^{\top}\Psi_{ij}d_{ij} - \tilde{d}_{ij} e_{ij}^{\top}\operatorname{sgn}(e_{ij}).
		\end{aligned}
	\end{equation*}
	
	Since $\Psi_{ij}^{\top}\Psi_{ij}=I_3$, it follows that
	\begin{equation*}
		\begin{aligned}
			e_{ij}^{\top}\Psi_{ij}d_{ij}\hspace{-0.05cm}-\hspace{-0.05cm}\tilde{d}_{ij} e_{ij}^{\top}\operatorname{sgn}(e_{ij}) \hspace{-0.05cm} \leq \hspace{-0.05cm} \|e_{ij}\|\|d_{ij}\|\hspace{-0.05cm}-\hspace{-0.05cm}\tilde{d}_{ij} e_{ij}^{\top}\operatorname{sgn}(e_{ij})\hspace{-0.05cm}
			\leq \hspace{-0.05cm} 0,
		\end{aligned}
	\end{equation*}
	where the last inequality follows from $\|d_{ij}\|\leq\tilde{d}_{ij}$ and $e_{ij}^{\top}\operatorname{sgn}(e_{ij})=\|e_{ij}\|_1\geq\|e_{ij}\|$.
	
	Therefore, it is easy to see that
	\begin{equation}\label{STEP2_EXT}
		\dot{W}_{ij} \leq -\|x_{ij}-\eta_{ij}\|^2+\|e_{ij}\|\|x_{ij}-\eta_{ij}\| -\gamma e_{ij}^{\top}e_{ij}.
	\end{equation}
If $\gamma>\frac{1}{4}$, $\dot{W}_{ij}(t)\leq 0$ and $W_{ij}(t)\geq 0$ for all $t\in[0,+\infty)$. By Lemma \ref{EL_property}\eqref{EL_PRO_1}, $e_{ij}^{\top}E_{ij}(x_{ij})e_{ij}\geq \underline{c}_{ij}^{E}\|e_{ij}\|^2$. Therefore, $x_{ij}-\eta_{ij}$, $e_{ij}$, $\mu_{ij}-\hat{\mu}_{ij}$, and $\tilde{d}_{ij}-\hat{d}_{ij}$ are bounded, and $W_{ij}(+\infty)$ exists and is finite. By combining \eqref{USV_controllerd} and \eqref{USV_controllere}, one has
	\begin{equation}\label{STEP2_EXT22}
		\dot{x}_{ij} + x_{ij} = e_{ij} + \dot{\eta}_{ij} + \eta_{ij}.
	\end{equation}
In light of \textit{Step 1}, both $\eta_{ij}$ and $\dot{\eta}_{ij}=\vartheta_{ij}$ are bounded. Moreover, since all terms in the dynamics of $\vartheta_{ij}$ are bounded, $\dot{\vartheta}_{ij}$ is also bounded. Since $e_{ij}$ is bounded, \eqref{STEP2_EXT22} represents a stable first-order system with a bounded input, which implies that $x_{ij}$ and $\dot{x}_{ij}$ are bounded. It then follows from \eqref{USV_controller} that $\dot{\hat{x}}_{ij}$ is bounded. Furthermore,
\[
\ddot{\hat{x}}_{ij}
=\dot{\vartheta}_{ij}-\dot{x}_{ij}+\dot{\eta}_{ij},
\]
and hence $\ddot{\hat{x}}_{ij}$ is bounded. Therefore, by the continuity of $C_{ij}$ and $\Upsilon_{ij}$, together with the uniform positive definiteness of $E_{ij}(x_{ij})$, it follows from \eqref{ELControl} that $\dot e_{ij}$ is essentially bounded along the Filippov solutions.

Integrating inequality \eqref{STEP2_EXT} over $[0,+\infty)$ yields
\begin{equation*}
	\left(\gamma-\frac{1}{4}\right)
	\int_{0}^{+\infty} e_{ij}^{\top}(\tau)e_{ij}(\tau)d\tau
	\leq W_{ij}(0)-W_{ij}(+\infty)<+\infty.
\end{equation*}

Since $e_{ij}$ is absolutely continuous and $\dot e_{ij}$ is essentially bounded along the Filippov solutions, $\|e_{ij}\|^2$ is uniformly continuous. Therefore, by applying Barbalat's Lemma \cite{khalilNonlinearSystems2002}, one can obtain $\lim_{t \to +\infty}\|e_{ij}(t)\|=0$. Combining \eqref{USV_controllerd} and \eqref{USV_controllere} yields
$\dot{x}_{ij}-\dot{\eta}_{ij}
=-(x_{ij}-\eta_{ij})+e_{ij}$.
Since $\lim_{t\to+\infty}\|e_{ij}(t)\|=0$, it follows from the stability of this first-order system that
$\lim_{t\to+\infty}\|x_{ij}(t)-\eta_{ij}(t)\|=0$.
\end{proof}

\begin{remark}
	Compared with \cite{YE2018266,Ye2019Unified,CHENG10156887,zou2023,nguyen2024,liu2024}, the proposed method accommodates both local and intra-coalition coupling inequality constraints for uncertain EL systems. The considered problem differs from the generalized NE setting in \cite{ZENG201920}, while the coupling inequalities require complementary slackness beyond the equality constraints in \cite{Deng9772719}. Furthermore, disturbance compensation requires no prescribed generation model, in contrast to \cite{Zhang8792368,LIU2026112603,Romano8727896}. However, constraint satisfaction is guaranteed only at the limiting NE, not during transients.
\end{remark}

\section{Numerical Example}\label{Sec_SIMULATION}
To set up the simulation environment, let $x^\text{line}_\text{R} = -1000$, $x^\text{line}_\text{B} = 1000$, $X_\text{min} = -1000$, $X_\text{max} = 1000$, $Y_\text{min} = -1000$, and $Y_\text{max} = 1000$. Each swarm has a command center located along its respective supply line. The red and blue command centers are located at \((-1000,0)\) and \((1000,0)\), respectively. Each swarm consists of six USVs. For all USVs, the parameters are shown in Table \ref{ch4_tab_usv} \cite{antonelli2018}.

\begin{table}[!htbp]
	\center
\setlength{\abovecaptionskip}{-0.1cm}
\caption{Parameters of the USV model}\label{ch4_tab_usv}
\tabcolsep 6.5pt
\begin{tabular}{crcrcr}
	\toprule
	Parameter & Value & Parameter & Value & Parameter & Value \\
	\midrule
	$\tilde{m}_{ij}$    & 23.800   & $N_\delta$  & 0.105   & $X_{\dot{\varrho}}$     & -2.000  \\
	$Y_{\dot{\delta}}$ & -10.000  & $N_{\dot{\delta}}$ & 0.000 & $x_g$     & 0.046   \\
	 $N_{\theta}$     & -1.900   & $X_\varrho$  & -0.723 & $I_z$   & 1.760  \\
	$Y_{\dot{\theta}}$  & 0.000   & $Y_{\theta}$  & 0.108   & $N_{\dot{\theta}}$     & -1.000   \\
	  $Y_{\delta}$ & -0.861 &&&&\\
	\bottomrule
\end{tabular}
\end{table}

The task allocation, initial positions, and logistical supply distances of the red and blue USVs are provided in Tables \ref{TableAssign_Red} and \ref{TableAssign_Blue}, respectively. For all USVs, the parameters are set to $\phi_{ij}(0)=0$ and $\nu_{ij}(0)=\mathbf{0}$. For the red intrusive USVs, the attack angles are specified as $\tilde{\phi}_{11}^{\text{usv}}=120^{\circ}$ and $\tilde{\phi}_{12}^{\text{usv}}=60^{\circ}$. For the blue intrusive USVs, the attack angles are specified as $\tilde{\phi}_{21}^{\text{usv}}=225^{\circ}$, $\tilde{\phi}_{22}^{\text{usv}}=270^{\circ}$, and $\tilde{\phi}_{23}^{\text{usv}}=315^{\circ}$.

\begin{table}[!ht]		
	\center
	\setlength{\abovecaptionskip}{-0.1cm}
	\caption{Task allocation, initial positions, and logistical supply distances of red USVs}\label{TableAssign_Red}
\tabcolsep 14pt
\begin{tabular}{cccc}
	\toprule
	Red & Task & Initial Position & $\tilde{g}_{1j}$ \\ \midrule
	1 & Attack $\text{TRG}_{2}$ &$[-700, -500]^\top$ & 500  \\
	2 & Attack $\text{TRG}_{2}$ &$[-700, -300]^\top$ & 500  \\
	3 & Intercept $\text{USV}_{21}$ &$[-700, -100]^\top$ & 300 \\
	4 & Intercept $\text{USV}_{22}$ &$[-700, 100]^\top$ & 300 \\
	5 & Intercept $\text{USV}_{23}$ &$[-700, 300]^\top$ & 400 \\
	6 & Intercept $\text{USV}_{23}$ &$[-700, 500]^\top$ & 500 \\\bottomrule
\end{tabular}
\end{table}

\begin{table}[!ht]		
	\center
	\setlength{\abovecaptionskip}{-0.1cm}
	\caption{Task allocation, initial positions, and logistical supply distances of blue USVs}\label{TableAssign_Blue}
	\tabcolsep 14pt
	\begin{tabular}{cccc}
		\toprule
Blue & Task & Initial Position & $\tilde{g}_{2j}$\\ \midrule
 1 & Attack $\text{TRG}_{1}$ &$[700, -500]^\top$ & 500\\
2 & Attack $\text{TRG}_{1}$ &$[700, -300]^\top$ & 500\\
3 & Attack $\text{TRG}_{1}$ &$[700, -100]^\top$ & 500\\
 4 & Intercept $\text{USV}_{11}$ &$[700, 100]^\top$ & 300\\
 5 & Intercept $\text{USV}_{11}$ &$[700, 300]^\top$ & 300\\
	 6 & Intercept $\text{USV}_{12}$ &$[700, 500]^\top$ & 400\\
		\bottomrule
	\end{tabular}
\end{table}

The communication topologies adopted by the red and blue USV swarms for gradient estimation are illustrated in Figs. \ref{ch5_USV_graphs}(a) and (b), respectively. Moreover, Fig. \ref{ch5_USV_graphs}(c) depicts the communication topology employed by all USVs for auxiliary variable estimation.  

\begin{figure}[!t]
	\centering
	\begin{tikzpicture}[every node/.style={
			circle, draw=none, minimum size=1.5mm,
			font=\footnotesize\bfseries, text=white
		}, scale=0.6]
		
		\begin{scope}[xshift=0cm,yshift=3.5cm,
			every node/.append style={draw=PalaceRed, fill=PalaceRed}]
			\node (1a) at (0,1.5) {1};
			\node (2a) at (2,1.5) {2};
			\node (3a) at (4,1.5) {3};
			\node (6a) at (0,0) {6};
			\node (5a) at (2,0) {5};
			\node (4a) at (4,0) {4};
			\draw[PalaceRed, thick] (1a) -- (2a);
			\draw[PalaceRed, thick] (1a) -- (6a);
			\draw[PalaceRed, thick] (2a) -- (3a);
			\draw[PalaceRed, thick] (3a) -- (6a);
			\draw[PalaceRed, thick] (3a) -- (4a);
			\draw[PalaceRed, thick] (4a) -- (5a);
			\draw[PalaceRed, thick] (5a) -- (6a);
			\node[font=\footnotesize, text=black, draw=none, fill=none, inner sep=0pt] at (2,-1) {\textnormal{(a)}};
		\end{scope}
		
		\begin{scope}[xshift=6cm,yshift=3.5cm,
			every node/.append style={draw=DeepSeaBlue, fill=DeepSeaBlue}]
			\node (1b) at (0,1.5) {1};
			\node (2b) at (2,1.5) {2};
			\node (3b) at (4,1.5) {3};
			\node (6b) at (0,0) {6};
			\node (5b) at (2,0) {5};
			\node (4b) at (4,0) {4};
			\draw[DeepSeaBlue, thick] (1b) -- (2b);
			\draw[DeepSeaBlue, thick] (1b) -- (4b);
			\draw[DeepSeaBlue, thick] (1b) -- (6b);
			\draw[DeepSeaBlue, thick] (2b) -- (3b);
			\draw[DeepSeaBlue, thick] (3b) -- (4b);
			\draw[DeepSeaBlue, thick] (4b) -- (5b);
			\draw[DeepSeaBlue, thick] (5b) -- (6b);
			\node[font=\footnotesize, text=black, draw=none, fill=none, inner sep=0pt] at (2,-1) {\textnormal{(b)}};
		\end{scope}
		
		\begin{scope}[xshift=0cm,yshift=0cm]
			\foreach \i/\x in {1/0,2/2,3/4,4/6,5/8,6/10} {
				\node[circle, draw=PalaceRed, fill=PalaceRed, minimum size=2mm,
				font=\footnotesize\bfseries, text=white] (N\i) at (\x,1.5) {\i};
			}
			
			\foreach \i/\x in {6/0,5/2,4/4,3/6,2/8,1/10} {
				\node[circle, draw=DeepSeaBlue, fill=DeepSeaBlue, minimum size=2mm,
				font=\footnotesize\bfseries, text=white] (M\i) at (\x,0) {\i};
			}
			
			\draw[ForestGreen, thick, ->] (N1) -- (N2);
			\draw[ForestGreen, thick, ->] (N2) -- (N3);
			\draw[ForestGreen, thick, ->] (N3) -- (N4);
			\draw[ForestGreen, thick, ->] (N4) -- (N5);
			\draw[ForestGreen, thick, ->] (N5) -- (N6);
			
			\draw[ForestGreen, thick, ->] (M1) -- (M2);
			\draw[ForestGreen, thick, ->] (M2) -- (M3);
			\draw[ForestGreen, thick, ->] (M3) -- (M4);
			\draw[ForestGreen, thick, ->] (M4) -- (M5);
			\draw[ForestGreen, thick, ->] (M5) -- (M6);
			
			\draw[ForestGreen, thick, ->] (M6) -- (N1);
			\draw[ForestGreen, thick, ->] (N6) -- (M1);
			
			\node[font=\footnotesize, text=black, draw=none, fill=none, inner sep=0pt] at (5,-1) {\textnormal{(c)}};
		\end{scope}
		
	\end{tikzpicture}
	\vspace*{-0.2cm}
	\caption{Communication graphs. (a) Communication graph $\mathcal{G}_1$ of the USVs in the red swarm. (b) Communication graph $\mathcal{G}_2$ of the USVs in the blue swarm. (c) Communication graph $\bar{\mathcal{G}}$ of all the USVs.}
	\label{ch5_USV_graphs}
\end{figure}

%


For $j \in \{1,2\}$, the parameters of \eqref{CostFunctionRed_a} are set to $\hat{l}^{R}_{j1} = 0.45 $, $\hat{l}^{R}_{j2} =4.7\times 10^{-6}$, $\hat{l}^{R}_{j3} = 0.95$, and $\hat{l}^{R}_{j4} = 0.41$. For $j \in \{3,4,5,6\}$, the parameters of \eqref{CostFunctionRed_b} are set to $\hat{l}^{R}_{j1} = 0.21$, $\hat{l}^{R}_{j2} = 0.41$, $\hat{l}^{R}_{j3} = 0.6$, $\hat{l}^{R}_{j4} = 0.14$, and $\hat{l}^{R}_{j5} = 0.43$. Furthermore, choose $\Delta_{11} = [ 0, 100 ]^{\top}$, 
$\Delta_{12} = [0,-100 ]^{\top}$,
$\Delta_{13} = [0,150 ]^{\top}$,
$\Delta_{14} = [0,50 ]^{\top}$,
$\Delta_{15} = [0,-50]^{\top}$, and
$\Delta_{16} = [ 0, -150 ]^{\top}$. Moreover, for $j \in \{1,2,3\}$, the parameters of \eqref{CostFunctionBlue_a} are set to $\hat{l}^{B}_{j1} = 0.45$, $\hat{l}^{B}_{j2} = 4.7\times 10^{-6}$, $\hat{l}^{B}_{j3} = 0.95$, and $\hat{l}^{B}_{j4} = 0.41$. For $j \in \{4,5,6\}$, the parameters of \eqref{CostFunctionBlue_b} are set to $\hat{l}^{B}_{j1} = 1.2\times 10^{-3}$, $\hat{l}^{B}_{j2} = 0.55$, $\hat{l}^{B}_{j3} = 0.6$, $\hat{l}^{B}_{j4} = 0.1$, and $\hat{l}^{B}_{j5} = 0.8$. The selected task weights satisfy Assumption \ref{Ass_WeightLMI}, since $\lambda_{\min}(\Xi_{\mathrm{g}})>0$ and the minimum eigenvalues of both matrices in \eqref{LMI_HeadingNE} are $15.6$. Besides, select
$\Delta_{21} = [0 , 100 ]^{\top}$, 
$\Delta_{22} = [0 , 0 ]^{\top}$, 
$\Delta_{23} = [0 , -100 ]^{\top}$,
$\Delta_{24} = [0 , 100]^{\top}$,
$\Delta_{25} =[0 , 0 ]^{\top}$, and
$\Delta_{26} =[ 0 , -100 ]^{\top}$.

For all $i \in \mathbb{N}$ and $j \in \mathbb{V}_i$, the initial conditions are given by $\hat{\mu}_{ij}(0) = \mathbf{0}$, $\hat{d}_{ij}(0) = 0$, $\hat{x}_{ij}(0) = \mathbf{0}$, $\eta_{ij}(0) = \mathbf{0}$, $\vartheta_{ij}(0) = \mathbf{0}$, $\omega_{ij}(0) = \mathbf{0}$, $\rho_{ij}(0) = \mathbf{0}$, $\lambda_{ij}(0) = \mathbf{0}$, $\xi_{ij}(0) = \mathbf{0}$, and $\zeta_{ij}(0) = \mathbf{0}$. For $p \in \bar{\mathbb{N}}$, $s_{p,-p}(0)$ is set to $\mathbf{0}$. Moreover, for all $i \in \mathbb{N}$ and $j \in \mathbb{V}_i$, the bounded disturbance is specified as $d_{ij} = \operatorname{col}\{2\sin(0.5t), 3\sin(0.3t), 0.1\sin(0.1t)\}$. 
The control gains in \eqref{USV_controller}, \eqref{AuxSys}, \eqref{LagMulti}, \eqref{TrackGrad}, and \eqref{EstAct} are chosen as $\gamma=2$, $\alpha=170$, $\beta=135000$, and $\kappa=1180$, satisfying the gain conditions of Theorem \ref{lemma_axuEqu}. The position trajectories and corresponding heading angles of all USVs are shown in Figs. \ref{ch5_USVResults}(a) and (b), respectively. For visualization, the heading angles in Fig. 3(b) are presented modulo $360^\circ$.

\begin{figure}[!htbp]
	\centering
	\subfigure[]{
		\includegraphics[width=0.43\textwidth]{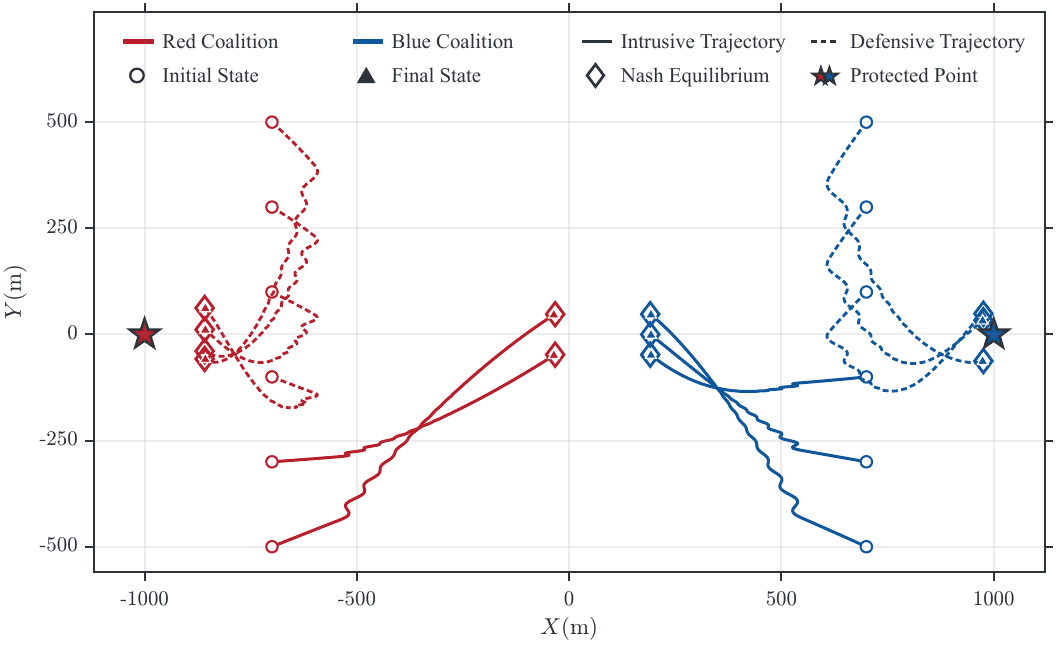}
	}\\
	\subfigure[]{
		\includegraphics[width=0.2\textwidth]{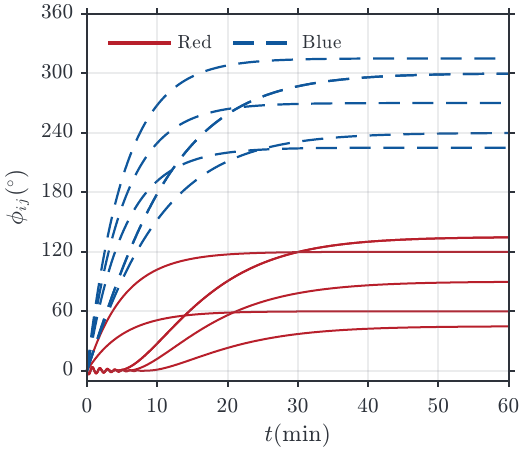}
	}
	\subfigure[]{
		\includegraphics[width=0.2\textwidth]{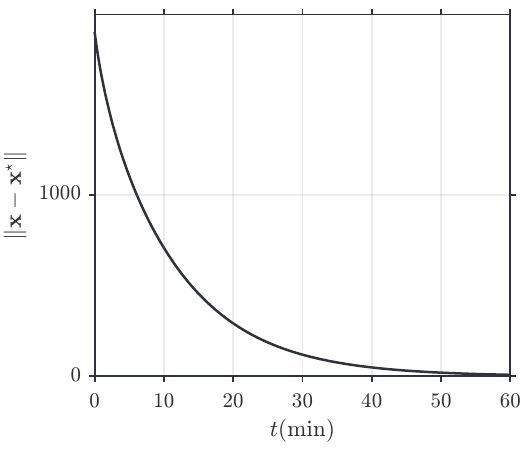}
	}
	\caption{State trajectories and convergence errors of all USVs.
		(a) Position trajectories.
		(b) Heading angles.
		(c) Errors relative to the NE.}
	\label{ch5_USVResults}
\end{figure}

The state trajectories approach the NE, as further illustrated by the decreasing error in Fig. \ref{ch5_USVResults}(c). The trajectories of the Lagrange multipliers for the local constraints are presented in Fig. \ref{ch5_Constraints}(a), while those for the coupling constraints are shown in Fig. \ref{ch5_Constraints}(b). As illustrated in Fig. \ref{ch5_USVResults}(a), the equilibrium positions of all USVs lie within the designated boundaries of the battlefield, suggesting that the Lagrange multipliers associated with the local constraints should asymptotically approach zero. This expectation is confirmed by the results depicted in Fig. \ref{ch5_Constraints}(a).

\begin{figure}[!htbp]
	\centering
	\subfigure[]{
		\includegraphics[width=0.2\textwidth]{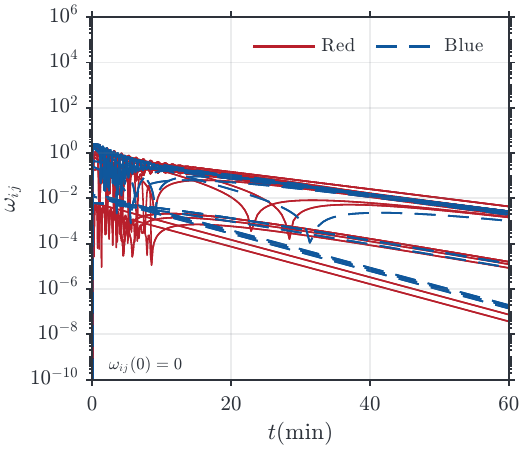}
	}
	\subfigure[]{
		\includegraphics[width=0.2\textwidth]{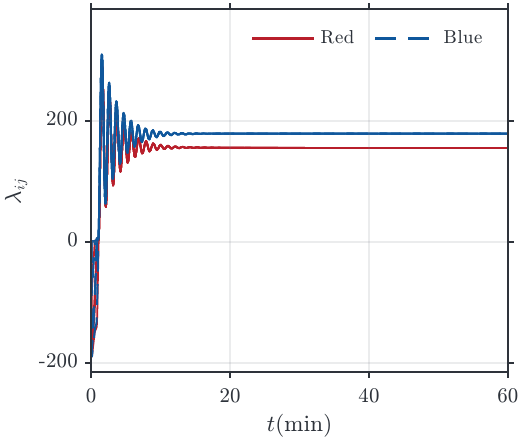}
	}
	\caption{Auxiliary multiplier variables corresponding to both constraints. (a) Local constraints. (b) Coupling constraints.}
	\label{ch5_Constraints}
\end{figure}

%

\section{Conclusion}\label{Sec_Conclu}
In this paper, we formulated USV swarm confrontation as a constrained coalition game and proposed a distributed NE seeking strategy for uncertain EL systems subject to disturbances with unknown bounds, while accounting for both local and coupling constraints. The game formulation captures cooperation within swarms and competition between swarms through tasks involving target approach, interception, and formation maintenance. In future work, we will develop distributed online NE seeking methods for time-varying coalition games.

\bibliographystyle{ieeetr}
\bibliography{Ref}

\begin{thebibliography}{10}

\bibitem{Rasmus2021Formation}
R.~Ringbäck, J.~Wei, E.~S. Erstorp, J.~Kuttenkeuler, T.~A. Johansen, and K.~H.
  Johansson, ``Multi-agent formation tracking for autonomous surface
  vehicles,'' {\em IEEE Trans. Control Syst. Technol.}, vol.~29, no.~6,
  pp.~2287--2298, 2021.

\bibitem{Wen2024escort}
G.~Wen, X.~Fang, H.~Shen, J.~Zhou, and D.~Zheng, ``Distributed leader escort
  control for multiple autonomous surface vessels: Utilizing signed graph to
  model interaction relationships,'' {\em IEEE/ASME Trans. Mechatron.},
  vol.~30, no.~4, pp.~2574--2585, 2025.

\bibitem{MA2021227}
S.~Ma, W.~Guo, R.~Song, and Y.~Liu, ``Unsupervised learning based coordinated
  multi-task allocation for unmanned surface vehicles,'' {\em Neurocomputing},
  vol.~420, pp.~227--245, 2021.

\bibitem{Yuwen2025}
C.~Yuwen, J.~Han, X.~Liu, and Z.~Zhen, ``Distributed decision-making of general
  linear systems in multi-coalition games and its application to {USV} swarm
  confrontation,'' {\em Int. J. Robust Nonlinear Control}, vol.~35, no.~3,
  pp.~1255--1268, 2025.

\bibitem{rubinstein2007}
J.~{von Neumann} and O.~Morgenstern, {\em Theory of Games and Economic
  Behavior}.
\newblock Princeton, USA: Princeton University Press, 1944.

\bibitem{weiss1999multiagent}
G.~Weiss, {\em Multiagent Systems: A Modern Approach to Distributed Artificial
  Intelligence}.
\newblock Cambridge, USA: MIT Press, 1999.

\bibitem{Nedic2009}
A.~Nedi{\'c} and A.~Ozdaglar, ``Distributed subgradient methods for multi-agent
  optimization,'' {\em IEEE Trans. Automat. Control}, vol.~54, no.~1,
  pp.~48--61, 2009.

\bibitem{SALEHISADAGHIANI2016209}
F.~Salehisadaghiani and L.~Pavel, ``Distributed {N}ash equilibrium seeking: A
  gossip-based algorithm,'' {\em Automatica}, vol.~72, pp.~209--216, 2016.

\bibitem{YE2018266}
M.~Ye, G.~Hu, and F.~L. Lewis, ``{N}ash equilibrium seeking for $n$-coalition
  noncooperative games,'' {\em Automatica}, vol.~95, pp.~266--272, 2018.

\bibitem{huang2022}
B.~Huang, C.~Yang, Z.~Meng, F.~Chen, and W.~Ren, ``Distributed nonlinear
  placement for multicluster systems: A time-varying {N}ash equilibrium-seeking
  approach,'' {\em IEEE Trans. Cybern.}, vol.~52, no.~11, pp.~11614--11623,
  2022.

\bibitem{zhou2024a}
J.~Zhou, G.~Wen, Y.~Lv, T.~Yang, and G.~Chen, ``Distributed resource allocation
  over multiple interacting coalitions: A game-theoretic approach,'' {\em IEEE
  Trans. Automat. Control}, vol.~69, no.~11, pp.~8128--8135, 2024.

\bibitem{Ye2019Unified}
M.~Ye, G.~Hu, F.~L. Lewis, and L.~Xie, ``A unified strategy for solution
  seeking in graphical $n$-coalition noncooperative games,'' {\em IEEE Trans.
  Automat. Control}, vol.~64, no.~11, pp.~4645--4652, 2019.

\bibitem{CHENG10156887}
Z.~Chen, X.~Nian, and S.~Li, ``{N}ash equilibrium seeking for incomplete
  cluster game in the cooperation-competition network,'' {\em IEEE Trans.
  Syst., Man, Cybern., Syst.}, vol.~53, no.~10, pp.~6542--6552, 2023.

\bibitem{zou2023}
Y.~Zou, Z.~Meng, and M.~V. Basin, ``Consensus-based distributed {N}ash
  equilibrium seeking strategies for constrained noncooperative games of
  clusters,'' {\em IEEE Trans. Syst., Man, Cybern., Syst.}, vol.~53, no.~12,
  pp.~7840--7851, 2023.

\bibitem{nguyen2024}
D.~T.~A. Nguyen, M.~Bianchi, F.~D{\"o}rfler, D.~T. Nguyen, and A.~Nedi{\'c},
  ``Constrained multi-cluster game: {{Distributed}} {N}ash equilibrium seeking
  over directed graphs,'' in {\em Proc. 63rd IEEE Conf. Decis. Control},
  pp.~4712--4719, 2024.

\bibitem{liu2024}
S.~Liu, D.~Wang, and Z.-G. Wu, ``Distributed multi-coalition games with general
  linear systems over {M}arkovian switching networks,'' {\em IEEE Trans. Netw.
  Sci. Eng.}, vol.~11, no.~6, pp.~6157--6168, 2024.

\bibitem{ZENG201920}
X.~Zeng, J.~Chen, S.~Liang, and Y.~Hong, ``Generalized {N}ash equilibrium
  seeking strategy for distributed nonsmooth multi-cluster game,'' {\em
  Automatica}, vol.~103, pp.~20--26, 2019.

\bibitem{Deng9772719}
Z.~Deng and Y.~Liu, ``{N}ash equilibrium seeking algorithm design for
  distributed nonsmooth multicluster games over weight-balanced digraphs,''
  {\em IEEE Trans. Neural Netw. Learn. Syst.}, vol.~34, no.~12,
  pp.~10802--10811, 2023.

\bibitem{add_3}
J.~Liu, Z.~Liu, Y.~Li, E.~Tian, and C.~Peng, ``Privacy-protected formation
  control for unmanned surface vehicles: A neural predictor-based
  noncooperative game framework,'' {\em IEEE Trans. Veh. Technol.}, vol.~75,
  no.~7, pp.~12137--12147, 2026.

\bibitem{add_4}
Y.~Jiang and Z.~Li, ``Fully distributed target encircling control of autonomous
  surface vehicles based on noncooperative games,'' {\em IEEE Trans. Intell.
  Veh.}, vol.~9, no.~4, pp.~4769--4779, 2024.

\bibitem{add_5}
W.~Wu, Y.~Zhang, Z.~Jia, J.-G. Lu, and W.~Zhang, ``Adaptive fault-tolerant
  fuzzy containment control for networked autonomous surface vehicles: A
  noncooperative game approach,'' {\em IEEE Trans. Fuzzy Syst.}, vol.~32,
  no.~7, pp.~4192--4204, 2024.

\bibitem{add_6}
W.~Wu, L.~Xu, R.~Liu, and W.~Zhang, ``Hierarchical optimal event-triggered
  framework for robust formation control of underactuated {ASVs},'' {\em
  Automatica}, vol.~189, p.~112983, 2026.

\bibitem{huang2024a}
Y.~Huang, Z.~Meng, and J.~Sun, ``Distributed {N}ash equilibrium seeking for
  multicluster aggregative game of {E}uler-{L}agrange systems with coupled
  constraints,'' {\em IEEE Trans. Cybern.}, vol.~54, no.~10, pp.~5672--5683,
  2024.

\bibitem{LIU2026112603}
L.~Liu, F.~Deng, J.~Chen, and M.~Lu, ``Distributed {N}ash equilibrium seeking
  for aggregative games of linear systems subject to unknown disturbances,''
  {\em Automatica}, vol.~183, p.~112603, 2026.

\bibitem{add_7}
Z.~Liu and J.~Huang, ``Distributed {N}ash equilibrium seeking for constrained
  aggregative games subject to disturbances with unknown bounds,'' {\em
  Automatica}, vol.~189, p.~113029, 2026.

\bibitem{Zhang8792368}
Y.~Zhang, S.~Liang, X.~Wang, and H.~Ji, ``Distributed {N}ash equilibrium
  seeking for aggregative games with nonlinear dynamics under external
  disturbances,'' {\em IEEE Trans. Cybern.}, vol.~50, no.~12, pp.~4876--4885,
  2020.

\bibitem{Romano8727896}
A.~R. Romano and L.~Pavel, ``Dynamic {NE} seeking for multi-integrator
  networked agents with disturbance rejection,'' {\em IEEE Trans. Control Netw.
  Syst.}, vol.~7, no.~1, pp.~129--139, 2020.

\bibitem{clarke2008nonsmooth}
F.~H. Clarke, Y.~S. Ledyaev, R.~J. Stern, and P.~R. Wolenski, {\em Nonsmooth
  Analysis and Control Theory}.
\newblock New York, NY, USA: Springer, 1998.

\bibitem{Ortega1998}
R.~Ortega, A.~Lor{\'\i}a, P.~J. Nicklasson, and H.~Sira-Ram{\'\i}rez, {\em
  Passivity-based Control of {Euler-Lagrange} Systems: Mechanical, Electrical
  and Electromechanical Applications}.
\newblock London, U.K.: Springer, 1998.

\bibitem{antonelli2018}
G.~Antonelli, {\em Underwater Robots}.
\newblock Cham, Switzerland: Springer, 2018.

\bibitem{Gadjov8354898}
D.~Gadjov and L.~Pavel, ``A passivity-based approach to {N}ash equilibrium
  seeking over networks,'' {\em IEEE Trans. Automat. Control}, vol.~64, no.~3,
  pp.~1077--1092, 2019.

\bibitem{moreau2004}
L.~Moreau, ``Stability of continuous-time distributed consensus algorithms,''
  in {\em Proc. 43rd IEEE Conf. Decis. Control}, pp.~3998--4003, 2004.

\bibitem{khalilNonlinearSystems2002}
H.~K. Khalil, {\em Nonlinear Systems}.
\newblock Upper Saddle River, NJ, USA: Prentice Hall, 2002.

\end{thebibliography}

\vfill

\end{document}